%% file: wfoodwebs.tex
\newif\ifConference
\newif\ifJournal
\newif\ifArXiv
\newif\ifAnonymous

\ArXivtrue
\ifArXiv
\Journaltrue
\fi

\ifConference
	
	\documentclass[a4paper,USenglish,cleveref,autoref,thm-restate,numberwithinsect,nameinlink]{llncs}

\include{commands.tex}

	\usepackage{hyperref}

	\ifAnonymous
	\author{UNKNOWN AUTHOR}{}{}{}{}
	\institute{UNKOWN INSTITUTE}
	\authorrunning{UNKNOWN}
	\else
	\author{
		Jannik Schestag\orcidID{0000-0001-7767-2970}
		\thanks{
			Funded by the Dutch Organisation for Scientific Research (NWO), 
			“Optimization for and with Machine Learning (OPTIMAL)”
			OCENW.GROOT.2019.015.
		}
	}
	\authorrunning{J. Schestag}
	\institute{
		TU Delft, The Netherlands.
		Electronic mail: \email{j.t.schestag@tudelft.nl}
	}
	\fi
\else
	
	\documentclass[a4paper,USenglish,cleveref,autoref,thm-restate,numberwithinsect,nameinlink]{lipics-v2021}

\include{commands.tex}
	\author{Jannik Schestag}{TU Delft, The Netherlands}{j.t.schestag@tudelft.nl}{https://orcid.org/0000-0001-7767-2970}{}
	

	
	\Copyright{Jannik Schestag}
	\hideLIPIcs
	
	\authorrunning{J. Schestag}
	
\fi

\usetikzlibrary{shadows}

\newcommand{\PDDlong}{\PROB{Optimizing PD with Dependencies}}

\newcommand{\PDD}{\PROB{\mbox{$\varepsilon$-PDD}}}
\newcommand{\sPDD}{\PROB{\mbox{$\varepsilon$-PDD$_{\text{s}}$}}}

\newcommand{\fPDD}{\PROB{\mbox{$1$-PDD}}}
\newcommand{\fsPDD}{\PROB{\mbox{$1$-PDD$_{\text{s}}$}}}

\newcommand{\hPDD}{\PROB{\mbox{$\nicefrac{1}{2}$-PDD}}}
\newcommand{\hsPDD}{\PROB{\mbox{$\nicefrac{1}{2}$-PDD$_{\text{s}}$}}}

\newcommand{\wPDD}{\PROB{\mbox{Weighted-PDD}}}
\newcommand{\rwPDDlong}{\PROB{\mbox{Restricted Weighted PDD}}}
\newcommand{\rwPDD}{\PROB{\mbox{rw-PDD}}}
\newcommand{\rwsPDD}{\PROB{\mbox{rw-PDD$_{\text{s}}$}}}

\newcommand{\gviable}{\mbox{$\gamma$-viable}\xspace}
\newcommand{\viable}{\mbox{$\epsilon$-viable}\xspace}
\newcommand{\fviable}{\mbox{$1$-viable}\xspace}

\DeclareMathOperator{\distclust}{cvd}

\DeclareMathOperator{\vc}{vc}

\newcommand{\todos}[2][]{\todo[#1,color=red!25!green!50]{ #2}}

\usepackage{hyphenat}
\title{Weighted Food Webs Make Computing Phylogenetic Diversity So Much Harder}

\titlerunning{Phylogenetic Diversity with Weighted Food Webs} 

\ifArXiv
\nolinenumbers
\fi

\begin{document}

\maketitle

\begin{abstract}
	\ifConference
\else
	Phylogenetic trees represent certain species and their likely ancestors.
	In such a tree,
	\fi
	present-day species are leaves and an edge from~$u$ to~$v$ indicates that $u$ is an ancestor of~$v$.
	Weights on these edges indicate the phylogenetic distance.
	The phylogenetic diversity (PD) of a set of species~$A$ is the total weight of edges that are on any path between the root of the phylogenetic tree and a species in~$A$.
	
	Selecting a small set of species that maximizes phylogenetic diversity for a given phylogenetic tree is an essential task in preservation planning, where limited resources naturally prevent saving all species.
	An optimal solution can be found with a greedy algorithm~[Steel, Systematic Biology, 2005; Pardi and Goldman, PLoS Genetics, 2005].
	However, when a food web representing predator-prey relationships is given, finding a set of species that optimizes phylogenetic diversity subject to the condition that each saved species should be able to find food among the preserved species is \NP-hard~[Spillner~et~al., IEEE/ACM, 2008].
	
	\looseness=-1
	We present a generalization of this problem, where, inspired by biological considerations, the food web has weighted edges to represent the importance of predator-prey relationships.
	We show that this version is \NP-hard even when both structures, the food web and the phylogenetic tree, are stars.
	To cope with this intractability, we proceed in two directions.
	Firstly, we study special cases where a species can only survive if a given fraction of its prey is preserved.
	Secondly, we analyze these problems through the lens of parameterized complexity.
	Our results include that finding a solution is fixed-parameter tractable with respect to the vertex cover number of the food web, assuming the phylogenetic tree is a star.
\end{abstract}

\section{Introduction}
The ongoing \emph{sixth mass extinction}~\cite{barnosky2011has,cowie2022sixth} presents a significant challenge to humanity.
\ifConference
\else
From an ethical standpoint, t%
\fi 
here is a moral imperative to preserve species~\cite{kopnina}; moreover, maintaining biodiversity is also critical for human well-being~\cite{rands2010biodiversity,cardinale2012biodiversity}.

However, conservation efforts are constrained by limited political will, funding, and other resources, making it impossible to protect every species that is on the edge of extinction.
As a result, strategic decisions are made about which species to prioritize.
To provide biological evidence on how relevant the protection of a certain set of species (taxa) is, biologists developed the \emph{phylogenetic diversity} (PD) measure~\cite{FAITH1992}.
Given a phylogenetic tree---a directed tree where today's species are leaves and edges describe how related a species is to it's genetic parent---the phylogenetic diversity of a set of species~$A$ is the total weight of edges on paths from the root to species in~$A$.
Although phylogenetic diversity is not a perfect proxy for biological diversity~\cite{karanth2019phylogenetic},
it is the best approach to capturing the number of unique features represented in a species set~\cite{faith2016pd} and
has become the most widely used biodiversity measures~\cite{vellend2011measuring}.
In the~\MPDlong~(\MPD) problem, one is given a phylogenetic tree and a budget~$k$, and the goal is to select~$k$ species that maximize phylogenetic diversity~\cite{FAITH1992}.
A greedy algorithm optimally solves \MPD~\cite{FAITH1992,steel,Pardi2005}.
\ifJournal
Various generalizations of \MPD have been defined and analyzed that make the problem more
realistic---for instance, allowing species-specific conservation costs as integers~\cite{hartmann,pardi07,GNAP}, or
selecting reservoirs wherein all species survive~\cite{moulton,bordewich2012budgeted}.
\fi

One important extension is the problem \PDDlong~(\PDD), introduced in~\cite{moulton}, where a food web encodes pre\-da\-tor-prey re\-la\-tion\-ships.
Here, the goal is to select~$k$ species that maximize phylogenetic diversity, with the constraint that each selected species must either be a food source of the ecological system or have at least one prey among the selected species.
Food webs are key ecological models that describe species’ roles in their environments and the flow of energy through ecosystems~\cite{pimm1982food}.
Introducing weights to these interactions—reflecting their ecological importance—gives further insight into the function of the system  and has become increasingly common for food webs~\cite{lieberman,girardin2024analysis,yang2024analysis}.
In fact, it has been noted that ``weighting ecological interactions is especially important in case of food webs''~\cite{scotti}.
However, \PDD assumes unweighted food webs, limiting its capacity to represent interaction significance.

\paragraph*{Our Contribution.}
\looseness=-1
We close this gap by introducing \wPDD, a generalization of \PDD in which the food web is edge-weighted.
We are tasked to select~$k$ species that maximize phylogenetic diversity under the constraint that each selected species is either a source or receives a total incoming weight of at least~1 from other selected species.
We prove that \wPDD is \NP-hard to solve, even on elementary instances, such as if the food web is a clique or a star.

To address this computational hardness, we pursue two directions.
First, we define and study the \rwPDDlong (\rwPDD) problem, where species require that a predefined fraction of their prey also be preserved.
This problem is a special case of \wPDD and generalizes the following.
\begin{itemize}
	\item \PDD: A selected species must have at least one preserved prey;
	\item \hPDD: At least half of the prey of a selected species must be preserved;
	\item \fPDD: All prey of selected species must be preserved.
\end{itemize}

Second, we perform a detailed analysis within the framework of parameterized complexity%
\ifJournal
.
In this field,
we ask whether instances~\Instance of a problem~$\Pi$, in which a problem-specific parameter~$p$ has value~$\kappa$, can be solved in~$f(\kappa)\cdot |\Instance|^{\Oh(1)}$ time (\FPT) or~$|\Instance|^{f(\kappa)}$ time (\XP), where~$f$ is a computable function and~$|\Instance|$ the size of the instance.
\else
, where
we ask whether instances~\Instance of a problem~$\Pi$
can be solved in~$f(\kappa)\cdot |\Instance|^{\Oh(1)}$ time (\FPT) or~$|\Instance|^{f(\kappa)}$ time (\XP), where~$f$ is a computable function, $\kappa$ is problem-specific parameter, and~$|\Instance|$ the size of the instance.
\fi
\Wh{1}-hardness with respect to~$p$ provides evidence that no \FPT-algorithm exists.

We examine \rwPDD and \fPDD with respect to parameters categorizing the structure of the food web.
We focus on the vertex cover number of instances of \rwPDD, where we provide an \XP-algorithm in the general case and, for the case that the phylogenetic tree is replaced with a vertex-weighting, called \rwsPDD, an \FPT-algorithm.
\ifJournal
We further present algorithms for \rwsPDD and \fsPDD that are \XP or \FPT with respect to the cluster vertex deletion number or the treewidth of the food web.
\fi
A comprehensive overview of the complexity results for \rwPDD and \rwsPDD with respect to the main structural parameters is provided in~\Cref{fig:1/2-results} and for \fPDD and \fsPDD in~\Cref{fig:1-results}.
\ifJournal

We observe some hardness results for \fPDD and \hPDD---which then also hold for \rwPDD---and show algorithms for~\rwPDD---which then also hold for the special cases.
\else
\fi

\paragraph*{Structure of the Paper.}
In the next section, we
\ifJournal
give definitions used throughout this paper and
\fi
prove the \NP-hardness of \wPDD and first observations.
In Sections~\ref{sec:structural-h} and~\ref{sec:structural-f}, we, respectively, analyze \rwPDD and \fPDD with respect to
\ifJournal
parameters that categorize the structure
\else
structural parameters
\fi
of the food web.
Finally, in \Cref{sec:discussion}, we discuss our results and present future research ideas.

\section{Preliminaries}
\label{sec:prelims}
\subsection{Definitions}
For a positive integer $a\in \mathbb{N}$,
by~$[a]$ we denote the set $\{1,2,\dots,a\}$, and
by~$[a]_0$ the set $\{0\}\cup [a]$.
For functions $f,f': A \to \mathbb{R}$, we define~$f(A') := \sum_{a \in A'} f(a)$ for subsets~$A'$ of~$A$, and we write~$f' \le f$ if~$f'(a) \le f(a)$ for all~$a \in A$.
For a condition~$\Phi$, the \emph{Kronecker delta}~$\delta_{\Phi}$
\ifConference
\else
takes the value~1
\fi
if~$\Phi$ holds and otherwise~\mbox{$\delta_{\Phi}$}
\ifConference
\else
takes the value~0
\fi.

\ifJournal
We write that some table entries store~$-\infty$.
In practice, this could be a large negative integer, for example~$-\PD(X)-1$.
\fi

We consider, unless stated otherwise, simple directed graphs~$G=(V,E)$ with vertex-set~$V(G):= V$ and edge-set~$E(G):= E$.
The \emph{underlying undirected graph of~$G$} is obtained by omitting edge directions.
\ifJournal
If the underlying undirected graph of~$G$ has a certain graph property~$\Pi$ of undirected graphs, we say that~$G$ has property~$\Pi$.
\fi
We write~$uv$ for directed edges from~$u$ to~$v$ and $\{u,v\}$ for an undirected edge between~$u$ and~$v$.
\ifJournal
The \emph{degree}~$\deg(v)$ of a vertex~$v$ is the number of edges incident with~$v$.
\fi
The \emph{in-degree}~$\deg^-(v)$ of a vertex~$v$ is the number of incoming edges at~$v$.
\ifJournal
The \emph{out-degree}~$\deg^+(v)$ is the number of outgoing edges of~$v$.
\fi
For a graph~$G$ and a vertex set~$V'\subseteq V(G)$, the \emph{subgraph of~$G$} induced by~$V'$ is denoted with~$G[V'] := (V', \{uv\in E(G) \mid u,v\in V'\})$.
\ifJournal
With~$G - V' := G[V\setminus V']$ we denote the graph obtained from~$G$ by removing~$V'$ and its incident edges.
\else
We define~$G - V' := G[V\setminus V']$.
\fi
A~\emph{star} with center~$v$ is a connected graph in which every edge is incident with~$v$.

\paragraph*{Phylogenetic Trees and Phylogenetic Diversity.}
A \emph{tree}~$T = (V,E)$ is a directed, connected, cycle-free graph, where the \emph{root}, often denoted with~$\rho$, is the only vertex with an in-degree of~zero, and each other vertex has an in-degree of~one.
Vertices that have an out-degree of~zero are called \emph{leaves}.

\looseness=-1
For a given set $X$, a \emph{phylogenetic~$X$-tree~$\Tree=(V,E,\w)$} is a tree~$T=(V,E)$ in which each non-leaf vertex has an out-degree of at least two, with an \emph{edge-weight} function~\mbox{$\w: E\to \mathbb{N}_{>0}$},
\ifJournal
and an implicit bijective labeling of the leaves with elements of~$X$.
Because of the bijective labeling, we interchangeably write leaf, taxon, and species.
\else
and~$X$ is the set of leaves.
\fi
In biological applications, $X$~is a set of \emph{taxa} (or species), all other vertices of~$\Tree$ correspond to biological ancestors of these taxa and edge weight~$\w(uv)$ describes the \emph{phylogenetic distance} between~$u$~and~$v$.
As~$u$ and~$v$ correspond to distinct, possibly extinct taxa, we assume this distance to be positive.
For an edge~$uv \in E$ in a tree, $v$ is a \emph{child} of~$u$.

\looseness=-1
Given a phylogenetic tree~\Tree and set $A \subseteq X$,
let $E_{\Tree}(A)$ denote the set of edges on a path to a leaf in~$A$.
The \emph{phylogenetic diversity}~$\PD(A)$ of $A$ is
\ifJournal
defined by 
\fi
\begin{equation}
	\label{eqn:PDdef}
	\PD(A) := \sum_{e \in E_{\Tree}(A)} \w(e).
\end{equation}
Informally, the phylogenetic diversity of a set $A$ is the total weight of edges
\ifJournal
on paths
\fi
to~$A$.

A degree-2 vertex~$v$ with incident edges~$uv$ and~$vw$ is~\emph{contracted} if an edge~$uw$ with weight~$\w(uv)+\w(vw)$ is added and~$v$ is removed.
A vertex~$v$~\emph{is identified with the root~$\rho$} if all children of~$v$ become children of the root~$\rho$ and~$v$ is removed.
Let~$A,B \subseteq X$ be taxa sets and
let~$E_{\Tree}^+(B)$ denote the set of edges~$uv$ for which~$B = \off(v)$.
(See \Cref{fig:contraction}.)
The~\emph{$(A,B)$-contraction} of a phylogenetic tree~\Tree results from applying
\ifJournal
these steps exhaustive after each other.
\else
exhaustively:
\fi
\begin{inparaenum}[1)]
	\item Remove all edges in~$E_{\Tree}(A)$ and in~$E_{\Tree}^+(B)$.
	\item Identify all vertices that became in-degree~zero vertices after Step~1 with the root.
	\item Contract all vertices with an in- and out-degree of 1.
\end{inparaenum}

We always consider $(A,B)$-contractions in the context of subtracting~$\PD(A)$ from the threshold of diversity.
Therefore,
intuitively, the $(A,B)$-contraction of a tree is the tree resulting from saving taxa in~$A$ and letting taxa in~$B$ die out.

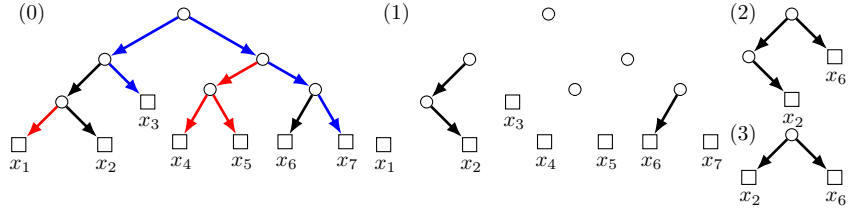
\begin{figure}[t]
	\tikzstyle{para}=[rectangle,draw=black,minimum height=.8cm,fill=gray!10,rounded corners=1mm, on grid]
	\centering
	\newcommand{\angledlabel}[3][]{\hspace{#2}\parbox{10ex}{{\bf #1}\newline\rotatebox{-15}{\textit{#3}}}}
	\begin{tikzpicture}[scale=.8, every node/.style={scale=.8}]
		\tikzstyle{sol}=[line width=1pt]
		\tikzstyle{nosol}=[dashed]

		\foreach \i/\red/\green in {0/red/blue,6/white/white}{
			\node[smallvertex] (root) at (\i,0) {};
			
			\nextnode{l}{root}{-150:1.5}{revarc, sol, \green}
			\nextnode{ll}{l}{-135:1}{revarc, sol}
			\nextnode[leaf, label=below:{\angledlabel[$x_1$]{8ex}{~}}]{lll}{ll}{-135:1}{revarc, sol, \red}
			\nextnode[leaf, label=below:{\angledlabel[$x_2$]{8ex}{~}}]{llr}{ll}{-45:1}{revarc, sol}
			\nextnode[leaf, label=below:{\angledlabel[$x_3$]{8ex}{~}}]{lr}{l}{-45:1}{revarc, sol, \green}
			
			\nextnode{r}{root}{-30:1.5}{revarc, sol, \green}
			\nextnode{rl}{r}{-150:1}{revarc, sol, \red}
			\nextnode{rr}{r}{-30:1}{revarc, sol, \green}
			\nextnode[leaf, label=below:{\angledlabel[$x_4$]{8ex}{~}}]{rll}{rl}{-120:1}{revarc, sol, \red}
			\nextnode[leaf, label=below:{\angledlabel[$x_5$]{8ex}{~}}]{rlr}{rl}{-60:1}{revarc, sol, \red}
			\nextnode[leaf, label=below:{\angledlabel[$x_6$]{8ex}{~}}]{rrl}{rr}{-120:1}{revarc, sol}
			\nextnode[leaf, label=below:{\angledlabel[$x_7$]{8ex}{~}}]{rrr}{rr}{-60:1}{revarc, sol, \green}
		}
		\node[textnode] at (-2.5,0) {(0)};
		\node[textnode] at (3.5,0) {(1)};
		
		\node[smallvertex] (root) at (10,0) {};
		\nextnode{l}{root}{-135:1}{revarc, sol}
		\nextnode[leaf, label=below:{\angledlabel[$x_2$]{8ex}{~}}]{lr}{l}{-45:1}{revarc, sol}
		\nextnode[leaf, label=below:{\angledlabel[$x_6$]{8ex}{~}}]{r}{root}{-45:1}{revarc, sol}
		\node[textnode] at (9.2,0) {(2)};
		
		\node[smallvertex] (root) at (10,-2) {};
		\nextnode[leaf, label=below:{\angledlabel[$x_2$]{8ex}{~}}]{l}{root}{-135:1}{revarc, sol}
		\nextnode[leaf, label=below:{\angledlabel[$x_6$]{8ex}{~}}]{r}{root}{-45:1}{revarc, sol}
		\node[textnode] at (9.2,-2) {(3)};
	\end{tikzpicture}
	\caption{
		(0): A hypothetical phylogenetic tree~\Tree.
		For~$A=\{x_3,x_7\}$ and~$B=\{x_1,x_4,x_5\}$, blue edges are in~$E_\Tree(A)$ and red edges are in~$E_\Tree^+(B)$.
		For~$i \in [3]$, ($i$) shows the~$(A,B)$-contraction of~\Tree after Step~$i$.
		To increase readability, edge weights are omitted.
	}
	\label{fig:contraction}
\end{figure}

\paragraph*{Food-Webs.}
For a set $X$ of taxa, a \emph{food web~$\Food=(X,E)$ on $X$} is a directed, acyclic graph with an edge-weight function~$\gamma: E \to (0,1]$.
For each edge $xy$, we say $x$ is \emph{prey} of $y$ and $y$ is a \emph{predator} of $x$.
The set of prey and predators of~$x$ are $\prey{x}$ and $\predators{x}$, respectively.
A taxon $x$ without prey is a \emph{source}.

For a food web~\Food, a set~$A \subseteq X$ of taxa
is~\emph{\gviable} if~$\sum_{e \in A_v} \gamma(e) \ge 1$ for each non-source~$v\in A$, where~$A_v$ is the set of edges~$uv \in E(\Food)$ with~$u\in A$.
In other words, each~$v\in A$ is either a source of~\Food, or the total weight of edges incoming from another vertex in~$A$ is at least~1.
If for each taxon all incoming edges have the same weight, then we say that~$\gamma$ is \emph{restricted}.
We observe that if $\gamma$ is restricted and~$A$ is~\gviable, then for any non-source~$v\in A$, at least~$\gamma_v := \lceil\gamma(uv)^{-1}\rceil$ prey of~$v$ are in~$A$, where~$uv$ is an arbitrary incoming edge of~$v$.
%
%

\paragraph*{Problem Definitions and Parameterizations.}
We define the following problem.

\problemdef{\wPDD}{
		A phylogenetic~$X$-tree~$\Tree$, a food web~$\Food$ on $X$ with edge-weights~$\gamma$, and integers~$k$ and~$D$.
	}{
		Is there a~\gviable set~$S\subseteq X$ of size at most~$k$ such that~$\PD(S)\ge D$?
	}

The set~$S$ is called a \emph{solution} of the instance.
We adopt the convention that $n$ is the number of taxa,~$|X|$,
and~$m$ is the number of edges of the food web, $|E(\Food)|$.
Observe that \Tree has $\Oh(n)$~edges.
In \textsc{\rwPDDlong (\rwPDD)},~$\gamma$ has to be restricted.
The problems \fPDD, \hPDD, and~\PDD are special cases of~\rwPDD where, respectively, $\gamma(e)$ is $1/\deg^-(v)$, $2/\deg^-(v)$, and~$1$ for each edge~$e$ incoming at~$v\in X$.
Thus, a taxon can be saved only if all, half, or at least one of its prey are also preserved.

In the respective special cases~\rwsPDD, \fsPDD, \hsPDD, and \sPDD, we require~\Tree to be a star.
It is noted that such an instance can be viewed as only containing a vertex-weighted food web and no phylogenetic tree~\cite{faller}.


For an instance of \rwPDD, we define~\Wmax
\ifJournal
to be the maximum number~$\gamma_x$ for~$x\in X$.
Informally, $\Wmax$ is
\else
as
\fi
the maximum number of prey of a taxon~$x$ that have to be saved so that~$x$ can be saved.
\ifJournal
We may assume that~$\Wmax \le k$ and~$\Wmax$ is at most the maximum in-degree in the food web.
\fi

%

\subsection{Related work}
\PDD has been defined by Moulton et al.~\cite{moulton}.
The conjecture that \PDD is \NP-hard~\cite{spillner} has been proven in~\cite{faller} even for the case that the food web is a directed tree---a spider graph to be more precise.
Further, \sPDD is \NP-hard even if the food web is bipartite~\cite{faller} but can be solved in polynomial time if the food web is a directed tree~\cite{faller}.
\PDD can be approximated with a constant factor if the longest path in the food web has a constant length~\cite{Dvorak}.

\looseness=-1
\ifJournal
\PDD has been studied within the framework of parameterized complexity~\cite{PDD}, and it
\else
It
\fi
has been shown that \PDD is \FPT when parameterized by the budget~$k$ plus the height of the phylogenetic tree~\cite{PDD}.
\ifJournal

Shortly after this paper was written, it was shown that 
\else
Further,
\fi
\fPDD and \hPDD are \Wh{1}-hard and in \XP, when parameterized by the budget~$k$ or the threshold of diversity~$D$~\cite{fPDDKernel}.
\hPDD is \Wh{1}-hard with respect to the treewidth of the food web, but \FPT when parameterized with the food web's node scanwidth~\cite{twvssw}.
None of the three problems, \fPDD, \hPDD, and \PDD, admits a polynomial kernel with respect to~$\vc+D$, where $\vc$ is the vertex cover of the food web~\cite{fPDDKernel}.

\subsection{Preliminary Observations}
\ifJournal
We start with some observations that we use throughout the paper.
\fi

\begin{lemma}
	\label{lem:solution-check}
	Given an instance $\Instance = (\Tree,\Food,k,D)$ of \wPDD and a set~$A\subseteq X$, one can check whether~$A$ is a solution of \Instance
	in $\Oh(n+m)$~time.
\end{lemma}
\begin{proof}
	We can compute whether~$\PD(A) \ge D$ in~$\Oh(n)$ time, by summing the weight of edges in~$E_\Tree(A)$.
	One can check~$|A| \le k$ in~$\Oh(k)$ time.
	To check whether~$A$ is \gviable, we need to iterate over the set of prey for each taxon and check the weight of edges coming from~$A$, which takes~$\Oh(m)$ time.
\end{proof}

\begin{lemma}
	\label{lem:solution-size}
	Let $\Instance = (\Tree,\Food,k,D)$ be a \yes-instance of \wPDD.
	A solution of size of exactly~$k$ exists, subject to~$k\le |X|$.
\end{lemma}
\begin{proof}
	Let~$S$ be a solution for \Instance with~$|S| < k$.
	Assume~$S \ne X$ and let~$x$ be a taxon in~$X\setminus S$ being a source or having all prey in~$S$.
	Such a taxon exists as~$\Food$ is a directed acyclic graph and has a topological order.
	Because~$S$ is \gviable, also~$S\cup \{x\}$ is~\gviable.
	Observe $\PD(S\cup \{x\}) \ge \PD(S)$ for each taxon~$x \in X$.
	So~$S\cup\{x\}$ is a solution and consequently, there is a solution of size~$k$. 
\end{proof}

\newcommand{\lemSubSolutions}[1]{
	\begin{lemma}[$\star$]
		#1
		Given a food web~\Food and sets of taxa~$R$ and~$Q$ such that no taxon of~$X\setminus R$ can reach a taxon of~$R$ and no taxon of~$Q$ can reach a taxon of~$X\setminus Q$.
		If~$S$ is \fviable in~$\Food - (R\cup Q)$, then~$S \cup R$ is \fviable in~\Food. 
	\end{lemma}
}
\lemSubSolutions{\label{lem:sub-solutions}}

\newcommand{\ProofLemSubSolutions}[1]{
	\begin{proof}
		Because no taxon of~$X\setminus R$ can reach a taxon of~$R$, we conclude~$\prey{x} \subseteq R$ for each~$x\in R$.
		Analogously, $\prey{x} \subseteq X \setminus Q$ for each~$x \in X\setminus Q$.
		
		Assume that~$S$ is \fviable in~$\Food - (R\cup Q)$.
		Because~$S \subseteq X \setminus (R\cup Q)$,
		we conclude~$\prey{x} \subseteq S \cup R$ for each~$x\in S$.
		This proves the lemma.
	\end{proof}
}

\begin{figure}[t]
	\centering
	\begin{tikzpicture}[scale=.8]
		\clip [draw] (0,0) circle [x radius=5cm, y radius=8mm];;
		
		\draw [fill=red!30!] (-5,.8) rectangle (-2.5,-.8);
		\draw [fill=yellow!30!] (-1.1,.8) rectangle (-2.5,-.8);
		\draw [fill=blue!30!] (5,.8) rectangle (3,-.8);
		
		\draw [ultra thick,arrows={-Stealth[length=8pt]}] (-.3,0) -- (2.2,0);
		
		\node at (-3.75,0) {$R$};
		\node at (-1.8,0) {$S$};
		\node at (4,0) {$Q$};
	\end{tikzpicture}
	\caption{
		An illustration of the \Cref{lem:sub-solutions}.
		Here, all edges are directed towards the right.
	}
	\label{fig:example_RSQ}
\end{figure}
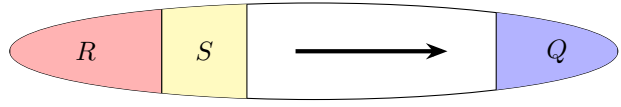

\subsection{Hardness of Weighted PDD}
\looseness=-1
Now, we prove that solving \wPDD is \NP-hard, even on instances that can be considered as containing only elementary information.

\newcommand{\thmWeighted}[1]{
	\begin{theorem}[$\star$]
		#1
		\wPDD is weakly \NP-hard in general and~\Wh{1}-hard when parameterized by the solution size~$k$, even if
		\begin{itemize}
			\item the phylogenetic tree is a star and the food web is a star, or
			\item the phylogenetic tree is a star and the food web is a clique.
		\end{itemize}
		These cases become strongly \NP-hard, if rationals are allowed as edge weights in the phylogenetic tree.
	\end{theorem}
}
\thmWeighted{\label{thm:weighted-PDD}}

Note that \PDD is strongly \NP-hard.
However, if the food web is a star or a clique, then solving the problem can be done in polynomial time, because after compulsorily saving the source, all taxa can be selected without further conditions and the instance can be reduced to \MPD and solved with Faith's greedy algorithm~\cite{FAITH1992}.
Consequently, this theorem shows \NP-hardness of cases that are computationally easy for \PDD and even for \rwPDD.

\newcommand{\CorrectnessThmWeighted}{
	\proofpara{Correctness}
	The reduction is computed in polynomial time.
	We only consider the correctness when~\Food is a star and omit the equivalent case of~\Food being a clique.
	
	Let~$A'$ be a solution of~$\Instance$ of size~$k$.
	We show that~$S:= A' \cup \{\star,\overline{a}\}$ is a solution of~$\Instance'$.
	It is~$\PD(S) = 2N + \nu(A') \ge 2N + D = D'$ and the size of~$S$ is clearly~$|A'|+2 = k+2$.
	It remains to show that~$S$ is \gviable.
	Since~$A'\cup\{\overline{a}\}$ are sources, it is sufficient to check that the incoming weight of~$\star$ is at least~1.
	It is
	\begin{eqnarray}
		\gamma(\overline{a} \star) + \sum_{a\in A'} \psi(a \star)
		&=& \frac{M + \sum_{a\in A'} M-c(a)}{M(k+1) - B}\\
		&=& \frac{(k+1)M - \sum_{a\in A'} c(a)}{M(k+1) - B}\\
		&\ge& \frac{(k+1)M-B}{M(k+1) - B} = 1.
	\end{eqnarray}
	Consequently,~$S$ is \gviable and a solution for~$\Instance'$.
	
	Conversely, let~$S$ be a solution for~$\Instance'$.
	For~$N$ big enough, we may assume~\mbox{$\star,\overline{a} \in S$}.
	We define~$A' := S \setminus \{\star,\overline{a}\}$ and show that~$A'$ is a solution for~$\Instance$.
	It is~$\nu(A') = \PD(S) - 2N \ge D$.
	Because~$S$ is \gviable, $\gamma(\overline{a} \star) + \sum_{a\in A'} \psi(a \star) \ge 1$.
	Further, we may assume by~\Cref{lem:solution-size} that~$|S| = k'$.
	Consequently,
	\begin{align}
		&& \frac{M + \sum_{a\in A'} M-c(a)}{M(k+1) - B} \ge& 1\\
		&\iff& M + \sum_{a\in A'} M-c(a) \ge& M(k+1) - B\\
		&\iff& \sum_{a\in A'} c(a) \le& B
	\end{align}
	Thus,~$A'$ is a solution of~\Instance.
}

\newcommand{\ProofThmWeighted}[1]{
	\begin{proof}
		\looseness=-1
		We reduce from \KP, in which a set of items~$A=\{a_1,\dots,a_n\}$,
		a cost-function~$c: A\to \mathbb{N}$, a value-function~$\nu: A\to \mathbb{N}$,
		and two integers~\mbox{$B,D\in \mathbb{N}$} are given.
		It is asked whether a set~$A' \subseteq A$ with~$c(A') \le B$ and~$\nu(A') \ge D$ exists.
		\KP is \NP-hard~\cite{karp} and~\Wh{1}-hard with respect to the solution size~$k$~\cite{downey}.
		Allowing rational costs and values makes \KP strongly \NP-hard~\cite{wojtczak2018strong}.
		
		Observe that after multiplying~$c(a)$ and~$B$ with~$k+1$ for each~$a\in A$ and adding~$k$ items of cost~1 and value~0, we may assume that if there is a solution, then there is also one of size~$k$.

		\proofpara{Reduction}
		Given an instance~$\Instance := (A=\{a_1,\dots,a_n\},c,\nu,B,D)$ of
		\KP, we construct an instance~$\Instance' := (\Tree,\Food,k',D')$ of \wPDD as follows.
		
		Define~$X := A \cup \{\star,\overline{a}\}$ and let~$N$ and~$M$ be big integers.
		Let \Tree be a star with root~$\rho$, leaves~$X$, and edge weights~$\w(\rho a) := \nu(a)$ for each~$a\in A$ and~$\w(\rho \star) := \w(\rho \overline{a}) := N$.
		Let~\Food contain edges~$a \star$ for each~$a\in A\cup \{\overline{a}\}$ of weight~$\gamma(a \star) := (M-c(a)) / (M(k+1) - B)$ and~$\gamma(\overline{a} \star) := M / (M(k+1) - B)$.
		
		As constructed so far,~\Food is a star.
		To obtain a clique, we add edges~$\overline{a} a_i$ and~$a_p a_q$, all of weight~1, for each~$i\in [n]$ and each combination~$1 \le p < q \le n$.

		Finally, we set~$k' := k+2$ and~$D' := 2N + D$.

		\looseness=-1
		\proofpara{Intuition}
		\ifJournal
		By the construction, it is ensured
		\else
		The construction ensures
		\fi
		that~$c(A') \le B$ if and only if~$A'' := A'\cup \{\star,\overline{a}\}$ is \gviable in~\Food and~$\nu(A') \ge D$ if and only if~$\PD(A'') \ge D'$ for any set~$A'\subseteq A$.

		#1
	\end{proof}
}
\ProofThmWeighted{\textit{The detailed correctness of this theorem is deferred to the appendix.}}

\looseness=-1
\section{Structural Parameters of the Food-Web for rw-PDD}
\label{sec:structural-h}
\ifConference
\else
In this section, we consider parameters that categorize the structure of the food web of an instance of \rwPDD.
\fi
A comprehensive overview of the complexity results for
\ifJournal
\rwPDD and \rwsPDD with respect to 
\fi
the main structural parameters are provided in~\Cref{fig:1/2-results}.
\ifJournal
We note that all
\else
All
\fi
three described \XP-algorithms are \FPT-algorithms if~$\Wmax$%
\ifJournal
---the maximum number of necessary prey to save for a taxon---%
\else
~
\fi
is added to the parameter.

\definecolor{r}{rgb}{1,0.5,0.5}  
\definecolor{o}{rgb}{1,0.65,0.3} 
\definecolor{a}{rgb}{1,0.85,0.2} 
\definecolor{g}{rgb}{0.7,1,0.7}  
\definecolor{w}{rgb}{1,1,1}      

\definecolor{grey}{rgb}{0.9453, 0.9453, 0.9453}
\definecolor{grey2}{rgb}{0.85, 0.85, 0.85}

\begin{figure}[t]
	\tikzstyle{para}=[rectangle,draw=black,minimum height=.8cm,fill=gray!10,rounded corners=1mm, on grid]
	
	\newcommand{\tworows}[2]{\begin{tabular}{c}{#1}\\{#2}\end{tabular}}
	\newcommand{\threerows}[3]{\begin{tabular}{c}{#1}\\{#2}\\{#3}\end{tabular}}
	\newcommand{\distto}[1]{\tworows{Distance to}{#1}}
	\newcommand{\disttoc}[2]{\threerows{Distance to}{#1}{#2}}
	
	\DeclareRobustCommand{\tikzdot}[1]{\tikz[baseline=-0.6ex]{\node[draw,fill=#1,inner sep=2pt,circle] at (0,0) {};}}
	\DeclareRobustCommand{\tikzdottc}[2]{\tikz[baseline=-0.6ex]{\node[draw,diagonal fill={#1}{#2},inner sep=2pt,circle] at (0,0) {};}}
	
	\tikzset{
		diagonal fill/.style 2 args={fill=#2, path picture={
				\fill[#1, sharp corners] (path picture bounding box.south west) -|
				(path picture bounding box.north east) -- cycle;}},
		reversed diagonal fill/.style 2 args={fill=#2, path picture={
				\fill[#1, sharp corners] (path picture bounding box.north west) |- 
				(path picture bounding box.south east) -- cycle;}}
	}
	\centering
	\begin{tikzpicture}[node distance=2*0.55cm and 3.7*0.45cm, every node/.style={scale=0.66}]
		\linespread{1}
		\node[para, diagonal fill=ga] (vc) {Minimum Vertex Cover};
		\node[para, diagonal fill=ar, xshift=38mm] (ml) [right=of vc] {Max Leaf \#};
		\node[para, xshift=-15mm, diagonal fill=aw] (dc) [left=of vc] {\distto{Clique}};
		
		\node[para, diagonal fill=ar, xshift=11mm] (dcc) [below= of dc] {\distto{Cluster}}
		edge (dc)
		edge[bend left=20] (vc);
		\node[para, diagonal fill=ar, xshift=8mm] (ddp) [below=of vc] {\distto{disjoint Paths}}
		edge (vc)
		edge (ml);
		\node[para, diagonal fill=ar] (fes) [below =of ml] {\tworows{Feedback}{Edge Set}}
		edge (ml);
		\node[para, xshift=2mm, diagonal fill=ar] (bw) [below right=of ml] {Bandwidth}
		edge (ml);
		\node[para, xshift=5mm, yshift=0mm, diagonal fill=ar] (td) [right=of ddp] {Treedepth}
		edge[bend right=20] (vc);
		
		\node[para, diagonal fill=ar] (fvs) [below= of ddp] {\tworows{Feedback}{Vertex Set}}
		edge (ddp)
		edge[bend right=5] (fes);
		\node[para, xshift=8mm, diagonal fill=ar] (pw) [below= of td] {Pathwidth}
		edge (ddp)
		edge (td)
		edge[bend right=8] (bw);
		
		\node[para, xshift=5mm, fill=r] (dbp) [below left= of fvs] {\distto{Bipartite}}
		edge (fvs);
		\node[para, yshift=0mm, diagonal fill=or] (tw) [below=of pw] {Treewidth}
		edge (fvs)
		edge (pw);

		\node[para, yshift=-20mm, xshift=-10mm, diagonal fill={blue!30}{grey2}] [below= of dc] {\tworows{{\rwPDD}\;\;\;\;\;\;\;\;\;\;\;}{\;\;\;\;\;\;\;\;\;\;\rwsPDD}};
	\end{tikzpicture}
	\caption{
		In this figure, the complexity of \rwPDD and \rwsPDD with respect to several structural parameters of the food web is presented.
		The complexity of \rwPDD is in the top left of each box, and the complexity of \rwsPDD is in the bottom right.
		A parameter~$p$ is marked
		in red~(\tikzdot{r}) if \rwPDD~/~\rwsPDD is \NP-hard for constant values of~$p$, or
		in amber~(\tikzdot{a}) or green~(\tikzdot{g}) if \rwsPDD~/~\rwsPDD admits an \XP-, or, respectively, an \FPT-algorithm with respect to~$p$.
		Classifying \rwPDD parameterized by distance to clique remains open.
		\rwsPDD with respect to treewidth is \Wh{1}-hard~\cite{twvssw} and in \XP.
		Two parameters~$p_1$ and~$p_2$ are connected with an edge if in every graph the parameter~$p_1$ further up is bounded by a function in $p_2$.
		A more in-depth look into the hierarchy of graph parameters can be found in~\cite{graphparameters}.
	}
	\label{fig:1/2-results}
\end{figure}
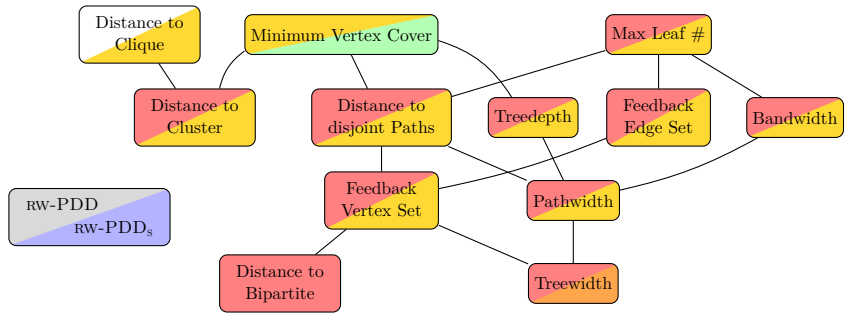

The hardness results are direct implications of results of~\cite{faller} or~\cite{PDD}.
\ifJournal
\PDD (in an undirected variant of the phylogenetic tree)
is \NP-hard even if
the phylogenetic tree has a height of~2 and
the food web is a directed tree~\cite{faller}---spider graphs in fact.
By a remark in~\cite{PDD},
in directed phylogenetic trees,
the \NP-hardness even holds when every connected component in the food web is a directed path of length~3.
Because in a directed path, every vertex has an in-degree of 1, these results thus generalize to \fPDD and \hPDD, as every taxon has at most one prey and then in all three variants of the problem, each non-source requires exactly their only prey to be saved, before it can be saved.

\begin{corollary}
	\label{cor:Faller}
	\fPDD and \hPDD remain \NP-hard on instances in which every connected component in the food web is a directed path of length~3.
	In such instances the maximum vertex degree in the food web is~2.
\end{corollary}

\PDD remains \NP-hard if the food web is an \emph{undirected} path and, therefore, the max-leaf number\footnote{The max-leaf \# of an undirected graph~$G$ is the maximum number of leaves any spanning tree of~$G$~has.} is~2~\cite{PDD}.
Using a similar approach, we show the following.

\newcommand{\corMaxLeaf}[1]{
	\begin{corollary}[$\star$]
		{#1}
		\fPDD and \hPDD are \NP-hard even if the food web is a path, and, therefore, the max-leaf number is~2.
	\end{corollary}
}
\corMaxLeaf{\label{cor:max-leaf-number}}
\newcommand{\ProofCorMaxLeaf}{
	\begin{proof}
		\proofpara{Reduction}
		Let $\Instance = (\Tree, \Food, k, D)$ be an instance of \PDD in which each connected component of \Food is a path of length three.
		Let $P^{(0)},P^{(1)},\dots,P^{(q)}$ be an arbitrary order of the connected components of \Food where $P^{(i)}$ contains the taxa $\{y_{i,0},y_{i,1},y_{i,2}\}$ and edges $y_{i,0}y_{i,1}$ and $y_{i,1}y_{i,2}$.
		Let~$M$ be a big constant.
		
		In the phylogenetic tree, we multiply every weight with~$M$.
		We add taxa~$p_1, \dots, p_{q-1}$ and make them children of the root in the food web with a weight~$\w(\rho p_i) = 1$ for each $i\in [q-1]$.
		In the food web, we add edges~$y_{i,2} p_i$ and~$y_{i+1,0} p_i$ for each $i\in [q-1]$.
		Finally, we set~$k' = k$ and set~$D' := D\cdot M$.
		
		\proofpara{Correctness}
		The reduction can be computed in polynomial time
		and it can be shown similarly as in \cite{PDD}, that this reduction is correct.
	\end{proof}
	\begin{figure}[t]
		\centering
		\begin{tikzpicture}[scale=0.7,every node/.style={scale=0.6}]
			\tikzstyle{txt}=[circle,fill=white,draw=white,inner sep=0pt]
			\tikzstyle{ndeg}=[circle,fill=blue,draw=black,inner sep=2.5pt]
			\tikzstyle{ndeo}=[circle,fill=orange,draw=black,inner sep=2.5pt]
			\tikzstyle{dot}=[circle,fill=white,draw=black,inner sep=1.5pt]
			
			\foreach \j/\name/\off in {0/0/0,2/1/0,4/2/0,9/q/-2} {
				\foreach \i in {0,...,2}
				\node[ndeg] (a\j\i) at (\j,\i) {};
				\foreach \i/\ii in {0/1,1/2}
				\draw[thick,arrows = {-Stealth[length=5pt]}] (a\j\i) to (a\j\ii);
				\node[txt] at (0.5+\j,0.5) {$P^{(\name)}$};
				
				\node[ndeo] (p\j) at (\j+1+\off,1) {};
			}
			
			\foreach \j in {0,2,4}{
				\draw[arrows = {-Stealth[length=3pt]}] (a\j2) .. controls +(right:2cm) and +(left:1cm) .. (p\j);
			}
			
			\foreach \j/\jj in {2/0,4/2,9/9}{
				\draw[arrows = {-Stealth[length=3pt]}] (a\j0) .. controls +(left:2cm) and +(right:1cm) .. (p\jj);
			}
			
			\foreach \j in {0,...,2}
			\node[dot] at (6+\j/2,1) {};
		\end{tikzpicture}
		\caption{
			An illustration of the food web in the reduction in the proof of~\Cref{cor:max-leaf-number}.
			The vertices of $X$ are blue and the new vertices are orange.
		}
		\label{fig:maxleaf}
	\end{figure}%
}
\fi

\subsection{Minimum Vertex Cover}
\label{sec:h-vc}
In this section, we parameterize \rwPDD with the minimum vertex cover number~($\vc$) of the food web~\Food.
\ifJournal
A vertex cover of~\Food is a set~$C \subseteq X$ such that~$u\in C$ or~$v\in C$ for each edge~$uv \in E(\Food)$.
\fi
We start with a useful pre-processing step.

\newcommand{\lemhHVC}[1]{
	\begin{lemma}[$\star$]
		#1
		Given an instance~$\Instance = (\Tree, \Food, k, D)$ of \rwPDD and a vertex cover~$C\subseteq X$ of~\Food of size~$\vc$,
		in~$\Oh(2^{\vc} \cdot (n+m))$ time, one can compute~$2^{\vc}$ instances~$\Instance_A = (\Tree_A, \Food_A, k_A, D_A)$ of \rwPDD, one for each~$A\subseteq C$, such that
		\Instance is a \yes-instance of \rwPDD, if and only if~$\Instance_A$ is a \yes-instance of~\rwPDD for some~$A\subseteq C$ and
		\ifJournal
		\begin{enumerate}
		\else
		\begin{inparaenum}[(1)]
		\fi
			\item the taxa in $A$ are children of the root of~$\Tree_A$,
			\item the height of $\Tree_A$ is at most the height of~$\Tree$,
			\item $\Tree_A$ contains~$\Oh(n)$ vertices,
			\item $u\not\in A$ and~$v\in A$ for each edge~$uv\in E(\Food_A)$,
			\item $\gamma$ remains unchanged on edges that are in both instances, and
			\item $A$ is a subset of each solution~$S$ of~$\Instance_A$.
		\ifJournal
		\end{enumerate}
		\else
		\end{inparaenum}
		\fi
	\end{lemma}
}
\lemhHVC{\label{lem:h-vc}}

Intuitively, $A' := C\setminus A$ and some taxa in~$X\setminus C$ can not survive, after fixing~$A$.
\ifJournal
In \rwPDD, we can prove this claim a bit easier by removing the condition that $\gamma$ has to remain unchanged on edges that are in both instances.
However, to make this lemma hold also for \hPDD, we prove this more challenging variant.
\fi
\newcommand{\ProofLemhHVC}{
	\begin{proof}
		\proofpara{Intuition}
		By the selection of~$A$, we know that~$A' := C\setminus A$ and some taxa in~$X\setminus C$ can not survive.
		We introduce a set~$Q$ that will mark the knowledge of how many prey have already been saved.
		
		\begin{figure}[t]
			\tikzstyle{bold}=[draw, line width=2pt]
			\tikzstyle{optional}=[dashed]
			\tikzstyle{path}=[decorate, decoration={snake, amplitude=.6mm}]
			
			\tikzstyle{small}=[inner sep=2pt]
			\tikzstyle{tiny}=[inner sep=1.7pt]
			\tikzstyle{textnode}=[inner sep=0pt]
			
			\tikzstyle{triangle}=[draw, regular polygon, regular polygon sides=3]
			\tikzstyle{vertex}=[circle, draw, fill=white]
			\tikzstyle{reti}=[vertex, fill=black]
			\tikzstyle{leaf}=[vertex, rectangle]
			\tikzstyle{leaf2}=[vertex, regular polygon, regular polygon sides=3]
			
			\tikzstyle{smallvertex}=[vertex, small]
			\tikzstyle{smallleaf}=[leaf, inner sep=3.3pt]
			\tikzstyle{smallleaf2}=[leaf2, inner sep=1.7pt]
			\tikzstyle{smalltriangle}=[triangle, inner sep=1.5pt]
			\tikzstyle{smallreti}=[reti, small]
			
			\tikzstyle{match}=[edge,line width=3pt]
			\tikzstyle{edge}=[draw,-]
			\tikzstyle{arc}=[draw,arrows={-Latex[length=6pt]}]
			\tikzstyle{boldarc}=[draw, bold, arrows={-Latex[length=10pt]}]
			\tikzstyle{revarc}=[draw, arrows={Latex[length=6pt]-}]
			\tikzstyle{boldrevarc}=[draw, bold, arrows={Latex[length=10pt]-}]
			
			\tikzstyle{rbox}=[rounded corners, color=red, fill=red!7]
			\tikzstyle{gbox}=[rounded corners, color=green!90!yellow!60!black!70!, fill=green!7]
			
			\centering
			\begin{tikzpicture}[scale=.9,every node/.style={scale=.9}]
				\draw[gbox] (1.25,2.75) rectangle (1.25+1.6666+.5,1.75);
				\node[textnode, color=green!90!yellow!60!black!70!] at (1,2.5) {$A$};
				\draw[rbox] (6.75,2.75) rectangle (6.75-1.6666-.5,1.75);
				\node[textnode, color=red] at (7,2.5) {$A'$};
				\draw[rbox] (3.75,.25) rectangle (6.25,-.75);
				\node[textnode, color=red] at (6.5,-.5) {$R$};

				\foreach \j in {1,...,4} {
					\node[smallvertex, label=above:$v_\j$] (m\j) at (\j*1.6666-0.1666,2) {};
				}
				\foreach \j in {1,...,7} {
					\node[smallvertex, label=below:$u_\j$] (i\j) at (\j,0) {};
				}

				\draw[arc] (m1) -- (m2);
				\draw[arc] (m3) -- (m2);
				\draw[arc] (m4) to[bend right=25] (m1);

				\foreach \m/\i in {1/1,1/3,2/5,3/4,3/5,3/3,4/5,4/6} 
				\draw[arc] (m\m) -- (i\i);

				\foreach \i/\m in {1/2,2/1,2/3,3/2,4/2,6/3,7/4} 
				\draw[arc] (i\i) -- (m\m);
			\end{tikzpicture}
			\begin{tikzpicture}[scale=.9,every node/.style={scale=.9}]
				\draw[gbox] (1.25,2.75) rectangle (1.25+1.6666+.5,1.75);
				\node[textnode, color=green!90!yellow!60!black!70!] at (1,2.5) {$A$};
				\draw[gbox] (.25,1.65) rectangle (-.75,.35);
				\node[textnode, color=green!90!yellow!60!black!70!] at (-.5,1.8) {$P$};

				\foreach \j in {1,2} {
					\node[smallvertex, label=above:$v_\j$] (m\j) at (\j*1.6666-0.1666,2) {};
				}
				\foreach \j/\label in {1/1,2/2,3/3,4/7} {
					\node[smallvertex, label=below:$u_\label$] (i\j) at (\j,0) {};
				}
				
				\foreach \j in {1,2} {
					\node[smallvertex, label=left:$p_\j$] (p\j) at (0,-.2+\j*.8) {};
				}

				\foreach \i/\m in {1/2,2/1,3/2}
				\draw[arc] (i\i) -- (m\m);
				
				\draw[arc] (p1) -- (m2);
				

				\draw (-1,-.75) -- ++ (0,3.5);
			\end{tikzpicture}
			\caption{\textbf{Left}: An example food web with indicated vertex sets.
				\textbf{Right}: The transformation that is done to this food web in the algorithm of \Cref{lem:h-vc}. (The phylogenetic tree is omitted.)}
			\label{fig:vc-preprocessing}
		\end{figure}
		
		\proofpara{Algorithm}
		For example, consider \Cref{fig:vc-preprocessing}.
		Iterate over subsets~$A\subseteq C$.
		We want~$A$ to be the set of taxa that need to survive and~$A' := C\setminus A$ to die out.
		Because~$C$ is a vertex cover,~$I := X\setminus C$ is an independent set.
		Let~$R \subseteq I$ be the set of taxa~$v \in I$ for which~$|\prey{v} \cap A| < |\prey{v} \cap A'|$ holds.
		
		Let~$P := \{v_1,\dots,v_{|A|}\}$ be a set of new taxa and let~$M$ and~$N$ be big integers.
		Compute the~$(A,A'\cup R)$-contraction~$\Tree'$ of~$\Tree$ and multiply each edge-weight with~$M$.
		Add~$A \cup P$ as new children to the root~$\rho$ of~$\Tree'$.
		Set the weight of edges~$\rho u$ to~$N$ for~$u\in A \cup P$.
		This completes the construction of~$\Tree_A$.
		
		To obtain~$\Food_A$, we add~$P$ to \Food.
		For each~$v\in A$, add~$|\prey{v} \cap A|$ edges~$wv$ to~$\Food_A$ with~$w\in P$, which all have the weight of all other edges incoming at~$v$.
		It does not matter which vertices~$w$ of~$P$ are chosen.
		Then remove~$A'$ with all incident edges from the food web. Remove all edges outgoing from~$A$.
		
		Finally, set~$k_A := k + |A|$ and~$D_A := N\cdot (D - \PD(A)) + 2M\cdot |A|$.

		\proofpara{Correctness}
		Conditions 1 to 4 hold by the construction.
		Observe that for~$M$ big enough, $A \cup P$ is a subset of every solution.
		It remains to show that \Instance is a \yes-instance of \rwPDD if and only if $\Instance_A$ is a \yes-instance of \rwPDD for some~$A\subseteq C$.
		
		Let \Instance be a \yes-instance of \rwPDD with solution~$S$.
		Define~$A := S \cap C$.
		Each vertex in~$I$ has all neighbors in~$C$.
		Each taxon in~$R$ has more prey in~$A' := C\setminus A$ than in~$A$.
		Therefore,~$R\cap S = \emptyset$.
		Prey~$u\in A$ of taxa~$v\in A$ are replaced with taxa~$u'\in P$.
		Therefore,~$S\cup P$ is \gviable in~$\Instance_A$, with a size of~$|S| + |P| = |S| + |A| \le k_A$, and~$\PDsub{\Tree_A}(S \cup P) = N \cdot (\PD(S) - \PD(A)) + M \cdot (|A| + |P|) \ge N \cdot (D - \PD(A)) + 2M \cdot |A| = D_A$.
		
		Conversely, let $\Instance_A$ is a \yes-instance of \rwPDD for~$A\subseteq C$ with solution~$S$.
		For a big enough~$M$, we can assume~$A\cup P \subseteq S$.
		Then, with an analogous argumentation,~$S' := S \setminus P$ is~\gviable in~$\Food$,~$|S'| \le k$ and~$\PD(S') \ge D$.

		\proofpara{Running Time}
		The iteration over the subsets of~$C$ takes~$2^{|C|}$ time.
		For a given set~$A$, we can compute~$R$ in~$\Oh(n+m)$ time.
		The tree~$\Tree_A$ and the food web~$\Food_A$ can be computed in~$\Oh(n+m)$ time.
	\end{proof}
}

\ifJournal
In the following, we use the result of \Cref{lem:h-vc}
and a dynamic programming algorithm over the phylogenetic tree
to prove that \rwPDD is \XP with respect to the food web's vertex cover
number and \FPT with respect to the vertex cover number plus~$\Wmax$.
Afterward, we prove with integer linear programming that \rwsPDD is \FPT with respect to the vertex cover number.
\else
\Cref{lem:h-vc} is used in both of the following theorems.
\fi
\newcommand{\thmhHVC}[1]{
	\begin{theorem}[$\star$]
		#1
		Let~$\Instance = (\Tree,\Food,k,D)$ be an instance of \rwPDD and~\mbox{$C \subseteq X$} a vertex cover of \Food of size~$\vc$,
		\Instance can be solved in~$\Oh((\Wmax+1)^{2\vc} \cdot (n+m)k)$~time.
	\end{theorem}
}
\thmhHVC{\label{thm:h-vc}}
\newcommand{\ProofThmhHVC}{
	\begin{proof}
		Apply \Cref{lem:h-vc}.
		Solve each of the instances~$\Instance_A = (\Tree_A, \Food_A, k_A, D_A)$ and return \yes, if any of them is a \yes-instance.
		Otherwise, if none of these is a \yes-instance, then return \no.
		
		To show how to solve~$\Instance_A$,
		we present a dynamic programming algorithm~$\DP$ over the tree~$\Tree_A$ which generalizes the one presented in~\cite{pardi07}.
		For any vertex~$v$ of the phylogenetic tree~$\Tree_A$, we define~$\Tree_A^{(v)}$ to be the subtree rooted at~$v$ and~$\off(v)$ to be the leaves in~$\Tree_A^{(v)}$.
		For a vertex~$v$ with children~$w_1,\dots,w_p$, we define~$\Tree_A^{(v,i)}$ for~$i\in [p]$ to be the subtree rooted at~$v$ where only the first~$i$ children of~$v$ are considered.
		Then,~$\off^{(i)}(v)$ are the leaves in~$\Tree_A^{(v,i)}$.
		
		\proofpara{Table Definition}
		We define~$\SSS_{v,f,k}$, for a vertex~$v$ of~$\Tree_A$, a function~$f: A \to \mathbb{N}_0$, and an integer~$k\in [k_A]_0$, to be the family of sets~$S\subseteq \off(v)$ which have a size of at most~$k$ and for which each~$a\in A$ has at least~$f(a)$ prey in~$S$.
		More formally,~$\SSS_{v,f,k} := \{ S \subseteq \off(v) \mid |S|\le k, |\prey{a} \cap S| \ge f(a) \forall a\in A \}$.
		For a vertex~$v$ with~$p$ children and an integer~$i\in [p]$, we define~$\SSS_{v,i,f,k}$ to be the subset of~$\SSS_{v,f,k}$, where~$S \subseteq \off^{(i)}(v)$.
		
		We define entry~$\DP[v,f,k]$ to be the maximum phylogenetic diversity of a set~$S \in \SSS_{v,f,k}$ in~$\Tree_A^{(v)}$.
		More formally,
		$$
		\DP[v,f,k] := \max\{ \PDsub{\Tree_A^{(v)}}(S) \mid S \in \SSS_{v,f,k} \}.
		$$
		Analogously,~$\DP'[v,i,f,k] := \max\{ \PDsub{\Tree_A^{(v,i)}}(S) \mid S \in \SSS_{v,i,f,k} \}$.
		
		\proofpara{Algorithm}
		As a base case, for each leaf~$x$, store~$\DP[x,f,k] = 0$ if~$k\ge 1$, and~$f(a) = 0$ for each~$a$ with~$x\not\in \prey{a}$, and~$f(a) \le 1$ for each~$a$ with~$x \in \prey{a}$.
		Otherwise, store~$\DP[x,f,k] = -\infty$.
		
		Let~$v$ be a vertex with children~$w_1,\dots,w_p$.
		Set~$\DP'[v,1,f,k] = \DP[v,f,k] + \delta_{k\ge 1} \cdot \w(vw_1)$.
		To compute further values, we use the following recurrences.
		\begin{eqnarray}
			\label{eqn:h-vc-XP}
			&& \DP'[v,i+1,f,k]\\
			\nonumber
			&=& \max_{k' \in [k]_0, f' \le f} \DP'[v,i,f-f',k-k'] + \DP[w_{i+1},f',k'] + \delta_{k'\ge 1} \cdot \w(vw_{i+1})
		\end{eqnarray}
		
		Finally, we set~$\DP[v,f,k] = \DP'[v,p,f,k]$.
		Let~$\rho$ be the root of~$\Tree_A$.
		Return \yes, if~$\DP[\rho,f,k_A] \ge D_A$, for some function~$f$ with~$f(a) \ge \gamma_a$ for each~$a\in A$.
		Otherwise, return \no.
		
		\proofpara{Correctness}
		Observe that for each~$S \in \SSS_{v,i+1,f,k}$, the set~$S' := S \cap \off(w_{i+1})$ is in~$\SSS_{w_{i+1},f',k'}$,
		where~$k' = |S'|$ and~$f'(a) := |S' \cap \prey{v}|$ for each~$a\in A$,
		and~$S \cap \off^{(i)}(v)$ is in~$\SSS_{v,i,f-f',k-k'}$.
		
		Conversely, for~$S_1 \in \SSS_{v,i,f_1,k_1}$ and~$S_2 \in \SSS_{w_{i+1},f_2,k_2}$,
		the set~$S_1 \cup S_2$ is in~$\SSS_{v,i+1,f_1+f_2,k_1+k_2}$.
		Then, the correctness of \Recc{eqn:h-vc-XP} follows from the observation that~$\PDsub{\Tree_A^{(v)}}(S) = \PDsub{\Tree_A^{(w_{i+1})}}(S) + \delta_{S \ne \emptyset} \cdot \w(vw_{i+1})$ for each~\mbox{$S \in \SSS_{w_{i+1},f,k}$}.

		The rest of the correctness follows intuitively.

		\proofpara{Running Time}
		As~$\Tree_A$ contains at most~$\Oh(n)$ vertices, both tables contain~$\Oh((\Wmax+1)^{\vc} \cdot nk)$ entries.
		
		The base cases can be checked in~$\Oh(m)$ time.
		\Recc{eqn:h-vc-XP} can be computed in~$\Oh((\Wmax+1)^{\vc} \cdot k)$ time.
		The overall running time is~$\Oh((\Wmax+1)^{2\vc} \cdot (n+m)k)$.
	\end{proof}
}

\ifJournal
In the following, we show how to,
after applying \Cref{lem:h-vc},
instances of~\rwsPDD can be reduced to instances of integer linear programming feasibility
(\ILPF), where the number of variables only depends on the size of the vertex cover of the food web.
\else
In the following, we show how to,
reduce instances of~\rwsPDD to instances of integer linear programming feasibility
(\ILPF) with~$2^{\vc}$ variables.
\fi
\ILPF on~$n$ variables can be solved using~$n^{2.5n+o(n)} \cdot |\Instance|$ arithmetic operations, where~$|\Instance|$ is the input length~\cite{frank,lenstra}.
Using a randomized algorithm even a running time of~$\log(2n)^{\Oh(n)}$ is possible~\cite{reis}.
It follows that \rwsPDD is \FPT when parameterized with the vertex cover number.
\begin{theorem}
	\label{thm:hs-vc}
	Let~$\Instance = (\Tree,\Food,k,D)$ be an instance of \rwsPDD and~$C \subseteq X$ a vertex cover of size~$\vc$.
	Then, \Instance can be solved in~$(\vc+1)^{\Oh(2^{\vc})} \cdot (n\log n + m)$ time.
\end{theorem}
\begin{figure}[t]
	\begin{minipage}[t]{20ex}
		\begin{tikzpicture}[domain=-.1:8, scale=.9, every node/.style={scale=.95}]
			\draw[very thin,color=gray] (-.1,-.1) grid (7.9,3.9);
			
			\draw[->] (-.2,0) -- (8.2,0) node[right] {$x$};
			\draw[->] (0,-.2) -- (0,4.2) node[above] {$f(x)$};
			
			\foreach \i in {1,...,3}
			\node at (-.3,\i) {$\i0$};
			\foreach \i in {1,...,7}
			\node at (\i,-.3) {$\i$};
			
			\draw[domain=-.1:3.54545454] plot (\x,1.1*\x) node[right] {$\Phi_M^{(1)}$};
			\draw[domain=-.1:4.5] plot (\x,.8*\x+.3) node[right] {$\Phi_M^{(2)}$};
			\draw[domain=-.1:6] plot (\x,.5*\x+.9) node[above] {$\Phi_M^{(3)} = \Phi_M^{(4)}$};
			\draw[domain=-.1:7.333333] plot (\x,.3*\x+1.7) node[above right] {$\Phi_M^{(5)}$};
			\draw plot (\x,.2*\x+2.2) node[right] {$\Phi_M^{(6)}$};
			\draw plot (\x,.1*\x+2.8) node[below] {$\Phi_M^{(7)}$};
			
			\draw[line width=1.5pt, color=red] (0,0) -- (1,1.1) -- (2,1.9) -- (3,2.4) -- (4,2.9) -- (5,3.2) -- (6,3.4) -- (7,3.5) {};
			\node[color=red] at (6.8,3.2) {$\Phi_M$};
		\end{tikzpicture}
	\end{minipage}\hfill
	\begin{minipage}[t]{20ex}
		\vspace{-4.5cm}
		\hspace{-1cm}
		\myrowcolsnum{2}
		\begin{tabular}{lcc}
			\toprule
			& ~~$\w(\rho v_i)$~~ & ~~$\sum$~~ \\
			\midrule
			$v_1$  & 11      & 11      \\
			$v_2$  & 8       & 19      \\
			$v_3$  & 5       & 24      \\
			$v_4$  & 5       & 29      \\
			$v_5$  & 3       & 32      \\
			$v_6$  & 2       & 34      \\
			$v_7$  & 1       & 35      \\
			\bottomrule
		\end{tabular}
	\end{minipage}
	\caption{An illustrative example of how to compute the function~$\Phi_M$ for values of~$\w(\rho v_i)$.}
	\label{fig:Phi}
\end{figure}
\begin{proof}
	%
	%
	%
	\proofpara{Algorithm and Correctness}
	Apply~\Cref{lem:h-vc} and iterate over the instances~$\Instance_A = (\Tree_A, \Food_A, k_A, D_A)$ of~\rwsPDD.
	We provide a reduction from~$\Instance_A$ to an instance of \ILPF with~$2^{|A|}$ variables.
	
	For subsets~$M$ of~$A$, define~$[M]_{\sim}$ as the set of taxa~$v\in X\setminus A$ that have~$M$ as predators.
	For each~$a\in A$, define~$A_a$ to be the family of sets~$S\subseteq A$ containing~$a$.
	We define an instance of \ILPF, with variables~$x_M$, upper bounded by~$q_M := |[M]_{\sim}|$, indicating how many taxa are chosen from~$[M]_{\sim}$.
	Recall that~$\gamma_a$ is the number of prey of a taxon~$a$ that have to be saved to save~$a\in A$.
	\begin{align}
		\label{eqn:ILPF-k}
		\sum_{M\subseteq A} x_M~\le~& k_A - |A|\\
		\label{eqn:ILPF-A}
		\sum_{M\in A_a} x_M~\ge~& \gamma_a & \forall a\in A\\
		\label{eqn:ILPF-D}
		\sum_{M\subseteq A} \Phi_M(x_M)~\ge~& D_A - \PDsub{\Tree_A}(A)\\
		\label{eqn:ILPF-x}
		\sum_{M\subseteq A} x_M~\le~& q_M
	\end{align}
	
	Recall, we have to save all taxa in~$A$, by~\Cref{lem:h-vc}.
	Inequality~(\ref{eqn:ILPF-k}) ensures that at most~$k_A$ taxa are saved.
	Inequality~(\ref{eqn:ILPF-A}) ensures that for each taxon~$a\in A$ the necessary number of prey are saved so that the solution is \gviable.
	Inequality~(\ref{eqn:ILPF-x}) provides the (logical) upper bound of~$x_M$.
	With $\Phi_M(x_M)$, the best phylogenetic diversity that can be achieved when~$x_M$ taxa are saved from~$M$ is given.
	Since all taxa in~$A$ have to be saved,~$D_A - \PDsub{\Tree_A}(A)$ diversity has to be contributed overall from the taxa~$X\setminus A$.
	Thus, Inequality~(\ref{eqn:ILPF-D}) ensures the diversity threshold is met.
	It remains to show how to compute~$\Phi_M(x_M)$.
	We do this with an approach similar to the one used to show that \KP is \FPT when parameterized by the number of numbers~\cite{etscheid}.
	An example is given in~\Cref{fig:Phi}.
	
	For each~$M\subseteq A$, order the taxa~$v_1, \dots, v_{q_M}$ of~$[M]_{\sim}$, such that~$\w(\rho v_i) \ge \w(\rho v_{i+1})$, for each~$i\in [q_M]$, where~$\rho$ is the root of~$\Tree_A$.
	For~$i\in [q_M]$, define linear functions~$\Phi_M^{(i)}$ with~$\Phi_M^{(i)}(i-1) = \sum_{j=1}^{i-1} \w(\rho v_j)$ and~$\Phi_M^{(i)}(i) = \sum_{j=1}^{i} \w(\rho v_j)$.
	Define~$\Phi_M(j) := \min_{i\in [q_M-1]} \Phi_M^{(i)}(j)$.
	This completes the algorithm.
	The correctness follows from the correct definition of the \ILPF instance.\todos{In Journal version: Maybe prove correctness of~$\Phi_M$}
	
	\proofpara{Running Time}
	The algorithm in \Cref{lem:h-vc} returns~$2^{\vc}$ instances in~$\Oh(2^{\vc} \cdot (n+m))$~time.
	The sets~$[M]_{\sim}$ can be computed in time~$\Oh(2^{\vc}+n+m)$ by an iteration over~$X$ and computing the predators.
	All functions~$\Phi_M$ are computed in~$\Oh(2^{\vc} \cdot n\log n)$ time.
	Then, the overall running time is dominated by the running time of \ILPF, which is~$\log(2\cdot 2^{\vc})^{\Oh(2^{\vc})} = (\vc+1)^{\Oh(2^{\vc})}$.
\end{proof}

\subsection{Distance to Cluster}
\label{sec:h-cluster}
In this section, we consider \rwPDD on instances where the food web is almost a cluster graph.
In a cluster graph, every connected component is a clique.
Cluster graphs generalize cliques and independent sets.

\ifJournal
The problem definitions of \PDD, \fPDD, and~\hPDD interact differently with cliques as food webs.
Let a clique with topological order~$x_0,\dots,x_\ell$ be given.
In \PDD, each clique is essentially an out-star, because once~$x_0$ (the source of the clique) is saved, each other vertex can be chosen without restrictions~\cite{PDD}.
In \fPDD, this property does not hold any longer.
But in this version, we can save taxon~$x_{i}$, after~$i$ taxa are saved from the clique.
Therefore, cliques essentially are equivalent to a path\todos{Quote f-PDD}.
In \hPDD, it becomes a bit trickier.
After saving~$x_0$, we are able to save taxon~$x_1$ and~$x_2$, because~$x_i$ has~$i$ incoming edges and it is therefore sufficient to save one prey for~$i\in \{1,2\}$.
Likewise, after saving~$i$ taxa, for any~$i$, we can save taxa~$x_2,\dots,x_{2i}$ without restrictions.
\fi

It remains open whether \hPDD---and therefore \rwPDD---can be solved in polynomial time
\ifJournal
on instances where the food web is a clique, while \PDD and \fPDD are almost trivial in this case.
In \PDD, it is sufficient to save the source, reduce to \MPD, and then run Faith's greedy~\cite{FAITH1992}.
In \fPDD, the topological order of the food web provides an order in which taxa are to be saved.
\else
on cliques.
In \PDD, it is sufficient to save the source and reduce to \MPD.
In \fPDD, we have to save taxa in topological order.
\fi

In the following, we observe that~\hPDD is \NP-hard if the food web is a cluster graph and show that \rwsPDD admits an \XP-algorithm when parameterized by the number of taxa that need to be removed to obtain a cluster.
The hardness result follows from
\ifConference
\else
\Cref{cor:Faller}.
We add one taxon for each connected component in the topological order between the two topmost vertices.
The edge weights of the phylogenetic tree are blown up by a big constant, and these new taxa are added as children of the root with a weight of 1.
Consider~\Cref{fig:cluster} for an illustration.
This finishes the reduction.
\begin{figure}[t]
	\centering
	\begin{tikzpicture}[scale=0.6,every node/.style={scale=0.6}]
		\tikzstyle{txt}=[circle,fill=white,draw=white,inner sep=0pt]
		\tikzstyle{ndeg}=[circle,fill=blue,draw=black,inner sep=2.5pt]
		\tikzstyle{ndeo}=[circle,fill=orange,draw=black,inner sep=2.5pt]
		\tikzstyle{dot}=[circle,fill=white,draw=black,inner sep=1.5pt]
		
		\foreach \j/\jj in {0/0,1.5/1,4.5/2,8/3,9.5/4,12.5/5} {
			\foreach \i in {0,...,2}
			\node[ndeg] (a\jj\i) at (\j,\i) {};
			\foreach \i/\ii in {0/1,1/2}
			\draw[thick,arrows = {-Stealth[length=3pt]}] (a\jj\i) to (a\jj\ii);
		}
		
		\foreach \j in {1,2,3,17,18,19}
		\node[dot] at (2+\j/2,1) {};
		
		\node at (6.25,1) {\scalebox{4}{$\leadsto$}};
		
		\foreach \j/\jj in {3/8.5,4/10,5/13} {
			\node[ndeg, fill=black] (a\j4) at (\jj,1.5) {};
			\draw[thick,arrows = {-Stealth[length=3pt]}] (a\j0) to[bend right=50] (a\j4);
			\draw[thick,arrows = {-Stealth[length=3pt]}] (a\j0) to[bend left=45] (a\j2);
			\draw[thick,arrows = {-Stealth[length=3pt]}] (a\j1) to (a\j4);
			\draw[thick,arrows = {-Stealth[length=3pt]}] (a\j4) to (a\j2);
		}
		
		%
	\end{tikzpicture}
	\caption{
		An illustration of the transformation done to the food web to prove~\Cref{cor:h-cluster}.
		Black vertices are new.
	}
	\label{fig:cluster}
\end{figure}%
\fi

\begin{corollary}
	\label{cor:h-cluster}
	\hPDD is \NP-hard, even if the food web is a cluster graph and each connected component contains four taxa.
\end{corollary}

In the following, we show that \rwsPDD is
\ifJournal
not only polynomial-time solvable on cluster graphs, but even
\fi
\XP with respect to the distance to
cluster%
\ifJournal
\footnote{In literature, distance to cluster is called cluster vertex deletion number ($\distclust$), also.}
\fi
 and \FPT when adding~$\Wmax$ to the parameter.

\newcommand{\thmHCluster}[1]{
	\begin{theorem}[$\star$]
		{#1}
		Let~$\Instance = (\Tree,\Food,k,D)$ be an instance of \rwsPDD and~\mbox{$M \subseteq X$} be a set of size~$\distclust$ such that~$\Food-M$ is a cluster graph.
		Then, \Instance can be solved in~$\Oh((\Wmax+1)^{2\distclust} \cdot n^2 k)$ time.
	\end{theorem}
}
\thmHCluster{\label{thm:h-cluster}}
\newcommand{\ProofThmHCluster}[1]{
	\begin{proof}
		\proofpara{Algorithm}
		Iterate over the subsets~$A$ of~$M$.
		We want taxa in~$A$ to survive and~$A' := M\setminus A$ to die out.
		Let~$A$ be fixed for the rest of the algorithm.
		For each~$x\in X$, compute~$x^{(A)} := \max\{0;\gamma_x - |\prey{x}\cap A|\}$.
		The number $x^{(A)}$ indicates, assuming that~$A$ is saved, how many prey of~$x$ in~$X\setminus M$ would need to be saved before~$x$ can be saved.
		Let~$C_1,\dots,C_t$ be the connected components in~$\Food - M$ and let~$x_{i,1},\dots,x_{i,|C_i|}$ be a topological order of~$C_i$, for each~$i\in [t]$.

		We define a dynamic programming algorithm with tables~$\DP$ and~$\DP_{(i)}$.
		For~$X' \subseteq X\setminus M$, $\ell \in [k]_0$, and a function~$f: A \to \mathbb{N}_0$,
		we define~$\SSS_{X',\ell,f}$ to be the family of sets~$S \subseteq X'$ such that~$|S|=\ell$, $a\in A$ has~$f(a)$ prey in~$S$, and~$x\in S$ has at least~$x^{(A)}$ prey in~$S$.
		In~$\DP[i,\ell,f]$, we store~$\max_{S \in \SSS_{X',\ell,f}} \PD(S)$, where~$X'$ is~$C_1 \cup \dots \cup C_i$, for~\mbox{$i\in [t]$}.
		In~$\DP_{(i)}[j,\ell,f]$, we store~$\max_{S \in \SSS_{X',\ell,f}} \PD(S)$, where~$X'$ is~$\{x_{i,1},\dots,x_{i,j}\}$, for~$i\in [t]$,~$j\in [|C_i|]$.
		Let~$\rho$ be the root of~\Tree.
		
		We define the function~$f_x$ as~$f_x(a) = \delta_{x \in \prey{a}}$ for each~$a\in A$.
		We indicate first how to compute~$\DP_{(i)}[j,\ell,f]$.
		We store~0 in~$\DP_{(i)}[1,0,f_0]$, where~$f_0$ maps all values to~0.
		As a base case, let~$\DP_{(i)}[1,\ell,f]$ store~$\w(\rho x_{i,1})$ if~$\ell=1$, $f = f_{x_{i,1}}$, and~$x_{i,1}^{(A)} \le 0$.
		Otherwise, store~$-\infty$.
		
		For~$j\in [|C_i|-1]$, we set $\DP_{(i)}[j+1,\ell,f]$ to~$\DP_{(i)}[j,\ell,f]$, or if $x_{i,j+1}^{(A)} \le \ell-1$, then to the maximum of~$\DP_{(i)}[j,\ell,f]$ and~$\DP_{(i)}[j,\ell-1,f - f_{x_{i,j+1}}] + \w(\rho x_{i,j+1})$.
		
		We set~$\DP[1,\ell,f]$ to~$\DP_{(1)}[|C_1|,\ell,f]$. For~$i\in [t-1]$, we use the recurrence
		\begin{eqnarray}
			\label{eqn:h-cluster}
			\DP[i+1,\ell,f] = \max_{\ell' \in [\ell], f' \le f} \{
			\DP[i,\ell',f];
			\DP_{(i+1)}[|C_{i+1}|,\ell - \ell',f - f']
			\}.
		\end{eqnarray}
		
		We return \yes, if~$\DP[t,\ell,f] \ge D - \PD(A)$ for some~$\ell\in [k]_0$ and some function~$f$ with~$f(a) \ge \gamma_a$ for each~$a\in A$.
		Otherwise, we continue with the next set~$A\subseteq M$.
		After the iteration over the subsets of~$M$, return \no.

		\proofpara{Correctness}
		We prove that~$\DP_{(i)}[j+1,\ell,f]$ for~$i\in [t]$,~$j\in [|C_i|-1]$ stores the right value, and omit the easier parts of the proof.
		Let~$S$ be a set of~$\SSS_{X_{j+1},\ell,f}$, where~$X_{j+1} := \{x_{i,1},\dots,x_{i,j+1}\}$.
		If~$x_{i,j+1} \not\in S$ then~$S \in \SSS_{X_{j},\ell,f}$.
		Otherwise, if~$x_{i,j+1} \in S$ then, by definition, $S$ contains at least $x_{i,j+1}^{(A)}$ prey of~$x_{i,j+1}$.
		Thus,~$x_{i,j+1}^{(A)} \le |S\setminus \{x_{i,j+1}\}| = \ell-1$.
		Then,~$S\setminus\{x_{i,j+1}\}$ is in $\SSS_{X_{j},\ell-1,f-f_{x_{i,j+1}}}$ and we conclude that $\DP_{(i)}[j+1,\ell,f] \le \max\{ \DP_{(i)}[j,\ell,f]; \DP_{(i)}[j,\ell-1,f - f_{x_{i,j+1}}] + \w(\rho x_{i,j+1}) \}$.
		
		Conversely, if~$S$ is in~$\SSS_{X_{j},\ell,f}$ then~$S$ is also in~$\SSS_{X_{j+1},\ell,f}$.
		Further, if~$S$ is in~$\SSS_{X_{j},\ell-1,f-f_{x_{i,j+1}}}$ and~$x_{i,j+1}^{(A)} \le \ell-1$, then $S\cup \{x_{i,j+1}\}$ is in~$\SSS_{X_{j+1},\ell,f}$.
		We conclude that~$\DP_{(i)}[j+1,\ell,f]$ stores the correct value.

		\proofpara{Running Time}
		The iteration over~$A$ takes~$\Oh(2^{\distclust})$ time.
		We note that it is sufficient to have~$f: A \to [\Wmax]_0$, where higher numbers map to~$\Wmax$.
		All tables together have~$\Oh((\Wmax+1)^{\distclust} \cdot nk)$ entries.
		
		Value can be computed with \Recc{eqn:h-cluster} in time~$\Oh((\Wmax+1)^{2\distclust} \cdot k^2)$.
		Any other step can be computed in time~$\Oh(n)$, such that the overall running time is~$\Oh((\Wmax+1)^{2\distclust} \cdot n^2 k)$.
	\end{proof}
}

\subsection{Treewidth}
\label{sec:h-treewidth}
Finally, we show that \rwsPDD is \XP with respect to the treewidth~$\tw$ of the food
web~\Food and \FPT when adding~$\Wmax$ to the parameter.
\ifJournal
Consequently,~\rwsPDD can be solved in polynomial time if the food web has a constant treewidth.
Common definitions of tree decompositions are found in~\cite{treedecom,cygan}.
\fi

\newcommand{\thmHTW}[1]{
	\begin{theorem}[$\star$]
		#1
		Given a nice tree-decomposition~$T$ of~$\Food=(V_\Food,E_\Food)$ with tree\-width~$\tw$,
		\rwsPDD can be solved in~$\Oh(\Wmax^{2\tw} \tw\cdot nk^2)$ time.
	\end{theorem}
}
\thmHTW{\label{thm:h-tw}}
\newcommand{\ProofThmHTW}{
	\begin{proof}
		Let $\Instance = (\Tree, \Food, k, D)$ be an instance of \rwsPDD.
		We define a dynamic programming algorithm with table~$\DP$ over the given tree-decomposition~$T$ of~$\Food=(V_\Food,E_\Food)$.
		
		For a node~$t\in T$, let $Q_t$ be the bag associated with~$t$ and let $V_t$ be the union of bags in the subtree of~$T$ rooted a~$t$.
		
		\proofpara{Table Definition}
		Given a bag~$t$, a set~$A\subseteq Q_t$, a function~$f: Q_t \to \mathbb{N}_0$, and an integer~$s$,
		a set~$Y\subseteq V_t$ is \emph{$(t,A,f,s)$-feasible}, if
		\begin{enumerate}
			\item[(T1)]\label{it:h-color}$A$ is the subset of~$Y$ in~$Q_t$; formally~$A = Y \cap Q_t$. 
			\item[(T2)]\label{it:h-prey}Each taxon $x\in Y\cap Q_t$ has~$f(x)$ prey in~$Y$;\\
			formally~$f(x) = |\prey{x} \cap Y|$ for all~$x \in Y\cap Q_t$.
			\item[(T2)]\label{it:h-safeness}Each taxon $x\in Y\setminus Q_t$ has at least~$\gamma_x$ prey in~$Y$;\\
			formally~$|\prey{x} \cap Y| \ge \gamma_x$ for all~$x \in Y\setminus Q_t$.
			\item[(T4)]\label{it:h-size}The size of $Y$ is $s$; formally~$s=|Y|$.
		\end{enumerate}
		
		Let~$\SSS_{t,A,f,s}$ be the set of $(t,A,f,s)$-feasible sets.
		In table entry~$\DP[t,A,f,s]$,
		store~$\max_{Y\in \SSS_{t,A,f,s}} \PD(Y)$.
		Let $r$ be the root of the tree-decomposition $T$.
		Then, $\DP[r,\emptyset,f_\emptyset,k]$ stores the maximum phylogenetic diversity of a \gviable, $k$-sized taxa set.
		Here,~$f_\emptyset$ is the ``function with an empty domain''.
		So, return~\yes if $\DP[r,\emptyset,f_\emptyset,k]\ge D$, and \no otherwise.

		\proofpara{Leaf Node}
		For a leaf~$t$ of~$T$ the bags~$Q_t$ and $V_t$ are empty. We store
		\begin{eqnarray} \label{tw:h-leaf}
			\DP[t,\emptyset,f_\emptyset,0] &=& 0.
		\end{eqnarray}
		For all other values, we store $\DP[t,R,G,B,s] = -\infty$.
		
		\Recc{tw:h-leaf} is correct by definition.

		\proofpara{Introduce Node}
		Let~$t$ be an \emph{introduce node}, that is,~$t$ has a single child~$t'$ with~$Q_t = Q_{t'} \cup \{v\}$.
		
		If~$v\not\in A$, store~$\DP[t,A,f,s] = \DP[t',A,f_{|Q_{t'}},s]$.
		
		If~$v\in A$ and $v$ has exactly $f(v)$ prey in~$A$,
		store~$\DP[t,A,f,s] = \DP[t',A\setminus \{v\},f',s] + \PD(v)$.
		Here,~$f'$ is defined on predators~$w\in \predators{v} \cap A$ of~$v$ as~$f'(w) := f(w)-1$, and~$f'(u) = f(u)$ for each~$u\in Q_t \setminus (\predators{v} \cap A)$.
		
		Otherwise, if~$v\in A$ and $|\prey{v}\cap A| \ne f(v)$, store~$\DP[t,A,f,s] = -\infty$.

		If we want $v$ to be saved, $f$ needs to store the number of prey that~$v$ has in~$A$.
		Further, $v$ counts into the number of prey for each predator of~$v$ in~$A$.

		\proofpara{Forget Node}
		Let~$t$ be a \emph{forget node}, that is,~$t$ has a single child~$t'$ and~$Q_t = Q_{t'} \setminus \{v\}$.
		We store
		\begin{eqnarray} \label{tw:h-forget}
			\DP[t,A,f,s] = \max &\{&
			\DP[t',A,f^{(0)},s];\\
			&&
			\max_{i \in \{\Wmax,\dots,|\prey{v}|\}}
			\DP[t',A\cup\{v\},f^{(i)},s]
			\}.
		\end{eqnarray}
		Here,~$f^{(i)}$ is the function~$A\cup\{v\} \to \mathbb{N}_0$ with~$f^{(i)}_{|A} = f$ and~$f^{(i)}(v) = i$.

		If~$v$ is being saved, by definition, we need to save at least~$\gamma_v$ of the prey of~$v$.
		Define sets~$\SSS_{t,A,f,s}^{v,i} := \{ Y\in \SSS_{t,A,f,s} \mid v\in Y, f(v) = i \}$ and~$\SSS_{t,A,f,s}^{-v} := \{ Y\in \SSS_{t,A,f,s} \mid v\not\in Y \}$.
		The correctness of \Recc{tw:h-forget} follows from the observation that~$\SSS_{t,A,f,s}^{v,i}$ for~$i\in \{\gamma_v,\dots,|\prey{v}|\}$ and~$\SSS_{t,A,f,s}^{v}$ are a disjoint union of~$\SSS_{t,A,f,s}$.

		\proofpara{Join Node}
		Let~$t$ be a \emph{join node}, that is,~$t$ has two children~$t_1$ and~$t_2$ with~$Q_t = Q_{t_1} = Q_{t_2}$.
		We store
		\begin{eqnarray} \label{tw:h-join}
			&& \DP[t,A,f,s]\\
			\nonumber
			&=& \max_{f_1,f_2,s' \in [s-|A|]_0} \DP[t_1,A,f_1,|A|+s'] + \DP[t_2,A,f_2,|A|+s-s'] - \PD(A).
		\end{eqnarray}
		Here, functions~$f_1$ and~$f_2$ hold~$f(v) = f_1(v) + f_2(v) - |\prey{v}\cap A|$ for each~$v\in Y$.
		
		The correctness of \Recc{tw:h-join} follows from the fact that there are no edges between~$V_{t_1} \setminus Q_t$ and~$V_{t_2} \setminus Q_t$.
		Because~\Tree is a star, we can simply add the phylogenetic diversities together.
		Further, $f_i$ counts the saved prey that are in~$V_{t_i}$ for~$i\in\{1,2\}$.
		Yet, prey in~$A$ is counted twice.

		\proofpara{Running Time}
		Instead of storing a subset of~$A\subseteq Q_t$ and a function~$f: Q_t \to \mathbb{N}$, we can store a function~$f: Q_t \to [\Wmax]_0 \cup \{\texttt{none}\}$, where we store~$f(v) \in \mathbb{N}$ if~$v\in A$ and~$f(v) = \texttt{none}$ if~$v \not\in A$.
		Higher values for~$f(v)$ can be mapped to~$\Wmax$.
		A tree decomposition contains $\Oh(n)$ nodes, thus the table contains $\Oh((\Wmax +1)^\tw \cdot nk)$ entries.
		Leaf, introduce, and forget nodes can be computed in time linear in~$|\prey{n}| \le n$ and~$\tw$.
		Observe that to compute the function~$f$ in a join node, it is sufficient to know~$A$,~$f_1$, and~$f_2$.
		Therefore, to compute all values of a join node, we iterate over~$s$,~$s'$,~$A$,~$f_1$, and~$f_2$ such that any join node can be computed in~$\Oh(\Wmax^{2\tw} \cdot k^2)$ time.
		Therefore, the overall running time is~$\Oh(\Wmax^{2\tw} \tw\cdot n\cdot k^2)$ time.
	\end{proof}
}

\section{Structural Parameters of the Food-Web for 1-PDD}
\label{sec:structural-f}
In this section, we analyze the complexity of \fPDD with respect to parameters that categorize the food web of an instance.
A detailed overview of these results is provided in \Cref{fig:1-results}.
It is somewhat remarkable that for all these parameterizations, \fPDD seemingly has the same tractability result as \PDD~\cite{PDD}.

\begin{figure}[t]
	\tikzstyle{para}=[rectangle,draw=black,minimum height=.8cm,fill=gray!10,rounded corners=1mm, on grid]
	
	\newcommand{\tworows}[2]{\begin{tabular}{c}{#1}\\{#2}\end{tabular}}
	\newcommand{\threerows}[3]{\begin{tabular}{c}{#1}\\{#2}\\{#3}\end{tabular}}
	\newcommand{\distto}[1]{\tworows{Distance to}{#1}}
	\newcommand{\disttoc}[2]{\threerows{Distance to}{#1}{#2}}
	
	\DeclareRobustCommand{\tikzdot}[1]{\tikz[baseline=-0.6ex]{\node[draw,fill=#1,inner sep=2pt,circle] at (0,0) {};}}
	\DeclareRobustCommand{\tikzdottc}[2]{\tikz[baseline=-0.6ex]{\node[draw,diagonal fill={#1}{#2},inner sep=2pt,circle] at (0,0) {};}}
	
	\tikzset{
		diagonal fill/.style 2 args={fill=#2, path picture={
				\fill[#1, sharp corners] (path picture bounding box.south west) -|
				(path picture bounding box.north east) -- cycle;}},
		reversed diagonal fill/.style 2 args={fill=#2, path picture={
				\fill[#1, sharp corners] (path picture bounding box.north west) |- 
				(path picture bounding box.south east) -- cycle;}}
	}
	\centering
	\begin{tikzpicture}[node distance=2*0.45cm and 3.7*0.38cm, every node/.style={scale=0.57}]
		\linespread{1}
		\node[para,fill=g] (vc) {Minimum Vertex Cover};
		\node[para, diagonal fill=gr, xshift=38mm] (ml) [right=of vc] {Max Leaf \#};
		\node[para, xshift=-25mm,fill=g] (dc) [left=of vc] {\distto{Clique}};
		
		\node[para, diagonal fill=gr, xshift=-4mm] (dcc) [below= of dc] {\distto{Cluster}}
		edge (dc)
		edge[bend left=10] (vc);
		\node[para, fill=g,xshift=26mm] (dcl) [below= of dc] {\distto{Co-Cluster}}
		edge (dc)
		edge (vc);
		\node[para, diagonal fill=gr, xshift=8mm] (ddp) [below=of vc] {\distto{disjoint Paths}}
		edge (vc)
		edge (ml);
		\node[para,diagonal fill=gr] (fes) [below =of ml] {\tworows{Feedback}{Edge Set}}
		edge (ml);
		\node[para, xshift=2mm, diagonal fill=gr] (bw) [below right=of ml] {Bandwidth}
		edge (ml);
		\node[para, xshift=5mm, yshift=0mm,diagonal fill=gr] (td) [right=of ddp] {Treedepth}
		edge[bend right=28] (vc);
		
		\node[para, diagonal fill=gr] (fvs) [below= of ddp] {\tworows{Feedback}{Vertex Set}}
		edge (ddp)
		edge[bend right=5] (fes);
		\node[para, diagonal fill=gr] (dcw) [right= of bw] {\tworows{Directed}{Cutwidth}};
		\node[para, diagonal fill=gr] (cw) [below= of bw] {Cutwidth}
		edge (bw)
		edge (dcw);
		\node[para, xshift=-30mm, diagonal fill=gr] (pw) [below= of cw] {Pathwidth}
		edge (ddp)
		edge (td)
		edge (cw);
		\node[para, diagonal fill=gr] (sw) [below right= of cw] {Scanwidth}
		edge (dcw);
		
		\node[para, xshift=5mm, fill=r] (dbp) [below left= of fvs] {\distto{Bipartite}}
		edge (fvs);
		\node[para, yshift=0mm, diagonal fill=gr] (tw) [below=of pw] {Treewidth}
		edge (fvs)
		edge (pw)
		edge[bend right=8] (sw);

		\node[para, yshift=-20mm, diagonal fill={yellow!50}{grey2}] [below= of dc] {\tworows{{\fPDD}\;\;\;\;\;\;\;\;\;\;\;}{\;\;\;\;\;\;\;\;\;\;\fsPDD}};
	\end{tikzpicture}
	\caption{
		This figure, similar to \Cref{fig:1/2-results}, shows the complexity of \fPDD and \fsPDD with respect to the main structural parameter of the food-web.
	}
	\label{fig:1-results}
\end{figure}

\subsection{Distance to Cluster}
\label{sec:f-cluster}
In this section, we consider how difficult \fPDD is to solve when the food web almost is a cluster graph.
Recall that in a cluster graph, every connected component is a clique.
%
%
In \fPDD, every clique is essentially a path, as every vertex that appears earlier in the topological orientation has to be saved first.
Consequently, with \cite{faller}, we can conclude the following for \fPDD.

\begin{corollary}
	\label{cor:f-cluster}
	\fPDD is \NP-hard, even if the food web is a cluster graph and each connected component contains 3 taxa.
\end{corollary}

Next, we show that \fsPDD is polynomial-time solvable when the food web is a cluster graph.
Afterward, we generalize this result and show that \fsPDD is \FPT when parameterized by the size of a given cluster vertex deletion set.

\begin{lemma}
	\label{lem:f-cluster}
	Instances of \fsPDD can be solved in~$\Oh((n+m)\cdot k^2)$ time, if the food web in the input is a cluster graph.
\end{lemma}
\begin{proof}
	%
	\proofpara{Algorithm}
	Let an instance~$\Instance := (\Tree, \Food, k, D)$ of \fsPDD be given, where~\Food is a cluster graph.
	Let~$C_1,\dots,C_q$ be the connected components of~\Food.
	For each~\mbox{$i\in [q]$}, the topological order of~$C_i$ directly indicates which set of taxa~$S_{i,j}$ will be saved if~$j\in [k]$ taxa can be saved from~$C_i$.
	Define~$\w_{i,j} := \PD(S_{i,j})$.
	
	Define a dynamic programming algorithm with table~$\DP$.
	In~$\DP[i,k']$, store the maximum phylogenetic diversity when~$k'$ taxa can be saved from~$C_1,\dots,C_i$.
	
	As a base case, for each~$j\in [\min\{k,|C_1|\}]_0$, store~$\DP[1,j] = \w_{1,j}$.
	
	To compute further values, we use the recurrence
	\begin{equation}
		\label{eqn:recurrence-f-cluster}
		\DP[i+1,j] :=
		\max_{\ell \in [j]_0}
		\DP[i,\ell]
		+
		\w_{i+1,j-\ell}.
	\end{equation}
	
	Return \yes if~$\DP[q,k] \ge D$.
	Otherwise, return \no.
	
	\proofpara{Correctness}
	Since the phylogenetic tree is a star, the only dependence of the taxa is given by the food web.
	Therefore, the sets~$S_{i,j}$ are well-defined.
	The rest of the proof is straight-forward.
	
	\proofpara{Running Time}
	By iterating over the edges, we can compute the in-degree of every vertex, which defines the topological order.
	Then, all values of~$\w_{i,j}$ can be computed in~$\Oh(n)$ time.
	The table $\DP$ has~$\Oh(q\cdot k)$ entries which can be computed in~$\Oh(k)$ time, each.
	Thus, the overall running time is~$\Oh(m\cdot k^2)$.
\end{proof}

\begin{theorem}
	\label{thm:f-cluster}
	Instances~$\Instance := (\Tree, \Food, k, D)$ of \fsPDD can be solved in~$\Oh(2^{|M|} \cdot (n+m)\cdot k^2))$ time if a set~$M\subseteq X$ is given such that~$\Food - M$ is a cluster graph.
\end{theorem}
\begin{proof}
	%
	\proofpara{Algorithm}
	Iterate over subsets~$Y\subseteq M$.
	We want that~$Y$ are the taxa in~$M$ that are being saved and~$M\setminus Y$ should die out.
	Let~$R_Y$ be the set of taxa which can reach~$Y$ in \Food and let~$Q_Y$ be the set of taxa which can be reached from~$M\setminus Y$ in \Food.
	If~$R_Y \cap Q_Y \ne \emptyset$, then continue with the next set~$Y$.
	Otherwise, compute whether~$\Instance' := (\Tree - (R_Y\cup Q_Y), \Food - (R_Y \cup Q_Y), k-|R_Y|, D-\PD(R_Y))$ is a \yes instance of \fsPDD with \Cref{lem:f-cluster} and return \yes if so.
	Otherwise, continue with the next set~$Y$.
	Return \no after the iteration.
	
	\proofpara{Correctness}
	Let~$S$ be a solution of~$\Instance$ and define~$Y := S\cap M$.
	By \Cref{lem:sub-solutions}, $R_Y \subseteq S$ and~$Q_Y \cap S = \emptyset$.
	We conclude that~$S\setminus R_Y$ is a solution of~$\Instance'$.
	As~$Y$ is considered in the iteration, the algorithm returns \yes.
	
	Conversely, assume that the algorithm returns \yes on~$Y$.
	Because~$Y$ is to be saved, each taxon which can reach~$Y$ needs to be saved.
	Similarly, each taxon that can be reached from~$M\setminus Y$ will go extinct when~$M\setminus Y$ does.
	Assume now that~$S$ is a solution for $\Instance'$.
	By \Cref{lem:sub-solutions}, $S \cup R_Y$ is valid in~\Food.
	Further, $|S \cup R_Y| = |S| + |R_Y| \le k$ and~$\PD(S \cup R_Y) = \PD(S) + \PD(R_Y) \ge D$.
	
	\proofpara{Running Time}
	For a given~$Y$, the sets~$R$ and~$Q$ can be computed in~$\Oh(n+m)$ time.
	By \Cref{lem:f-cluster}, we can compute a solution for~$\Instance''$ in~$\Oh((n+m)\cdot k^2)$ time.
\end{proof}

\subsection{Distance to Co-Cluster}
\label{sec:f-co-cluster}
Now, we show that \fPDD is \FPT with respect to the distance to co-cluster.
Recall, a co-cluster graph is the complement of a cluster graph.
Similar as in the last section, we show that \fPDD is polynomial-time solvable on co-clusters,~first.

\newcommand{\lemfCC}[1]{
	\begin{lemma}[$\star$]
		#1
		Instances of \fPDD can be solved in~$\Oh(nk\cdot (n+m))$ time, if the food web in the input is a co-cluster graph.
	\end{lemma}
}
\lemfCC{\label{lem:f-co-cluster}}
\newcommand{\CorrectnessLemfCC}{
	\proofpara{Correctness}
	Let~$S$ be a solution for~\Instance and consider the computed topology.
	Let~$x_i$ be the taxon of~$X\setminus S$ such that $A_i \subseteq S$.
	As~$x_i \not\in X\setminus S$ and $S$ is \fviable if and only if $X_{\le x} \subseteq S$ for each~$x\in S$~\cite{fPDDKernel}, $Q_i \cap S = \emptyset$. 
	Define~$S' := S \setminus A_i \subseteq X_i$ and observe~$|S'| = |S| - |A_i| \le k - |A_i|$ and $\PDsub{\Tree_i}(S') = \PD(S) - \PD(A_i) \ge D - \PD(A_i)$.
	Thus, $S'$ is a solution for~$\Instance_i$ and the algorithm returns \yes.
	
	Conversely, if there is a taxon~$x_i$ such that~$\Instance_i$ is a \yes-instance of \MPD with solution~$S_i$, then by analogous argument, $S_i \cup A_i$ is a solution for \Instance.
	
	\proofpara{Running Time}
	For each taxon~$x_i$, the sets~$A_i$ and~$Q_i$ can be computed in time~$\Oh(n+m)$.
	Faith's Algorithm for computing \MPD takes $\Oh(n \cdot k)$ time~\cite{steel,Pardi2005}.
	So, the overall running time is~$\Oh(n\cdot (n+m)\cdot k)$.
}

\newcommand{\ProofLemfCC}[1]{
	\begin{proof}
		%
		\looseness=-1
		\proofpara{Algorithm}
		Let an instance~$\Instance := (\Tree, \Food, k, D)$ of \fPDD be given, where~\Food is a co-cluster graph.
		Compute a topological order~$x_1,\dots,x_n$ of~\Food.
		Iterate over taxa~$x_i \in X$.
		We want~$x_i$ to be the first taxon to die out.
		By definition, the set~$A_i = \{x_1,\dots,x_{i-1}\}$ survives and the set~$Q_i$ of taxa reachable from~$x_i$\todos{Define as containing~$x_i$.} dies out.
		Observe that~$X_i := X \setminus (A_i \cup Q_i)$ are not neighbors of~$x_i$ in~\Food and so, as~\Food is a co-cluster,~$\Food[X_i]$ is an independent set.
		Let~$\Tree_i$ be the~$(A_i,Q_i)$-contraction of~$\Tree$.
		
		Return \yes, if~$\Instance_i := (\Tree_i,k-|A_i|,D-\PD(A_i))$ is a \yes instance of \MPD.
		Otherwise, continue with the next taxon.
		After the iteration, return \no.
		
		#1
	\end{proof}
}
\ProofLemfCC{\textit{The detailed correctness and running time is deferred to the appendix.}}

\looseness=-1
\begin{theorem}
	\label{thm:f-co-cluster}
	Instances~$\Instance := (\Tree, \Food, k, D)$ of \fPDD can be solved in~$\Oh(2^{|M|} \cdot nk\cdot (n+m))$ time if a set~$M\subseteq X$ is given such that~$\Food - M$ is a co-cluster graph.
\end{theorem}
\Cref{thm:f-co-cluster} is proven similar to \Cref{thm:f-cluster}.
We iterate over subsets~$Y$ of~$M$ and want that~$Y$ are the taxa that are surviving, while~$M \setminus Y$ do not survive.
After removing the taxa which can reach~$Y$ or which can be reached from~$M\setminus Y$, the food web is a co-cluster and a solution can be found with~\Cref{lem:f-co-cluster}.
%
%
%

\subsection{Treewidth}
\label{sec:f-treewidth}
In the following, we show that \fsPDD is \FPT with respect to the treewidth~$\tw$ of~\Food.
We use a coloring on the vertices to indicate whether a taxon is saved or not.
This approach is similar to the one used in~\cite{PDD}, to show that \sPDD is \FPT when parameterized with $\tw$.
\ifJournal
Since \PDD and \fPDD are \NP-hard even if the food web is a directed tree, not much hope remains that these algorithms can be generalized.
We do not define tree-decompositions.
\fi
Common definitions can be found in~\cite{treedecom,cygan}.

\newcommand{\thmfTW}[1]{
	\begin{theorem}[$\star$]
		#1
		Instances~$\Instance := (\Tree, \Food, k, D)$ of \fsPDD can be solved in~$\Oh(2^\tw \tw \cdot nk^2)$ time if  a nice tree-decomposition~$T$ of~%
		\ifJournal
		$\Food=(V_\Food,E_\Food)$
		\else
		$\Food$
		\fi
		with tree\-width~$\tw$ is given.
	\end{theorem}
}
\thmfTW{\label{thm:f-tw}}
\newcommand{\ProofThmfTW}{
	\begin{proof}
		Let $\Instance = (\Tree, \Food, k, D)$ be an instance of \fsPDD.
		We define a dynamic programming algorithm with table~$\DP$ over the given tree-decomposition~$T$ of~$\Food=(V_\Food,E_\Food)$.
		
		For a node~$t\in T$, let $Q_t$ be the bag associated with~$t$ and let $V_t$ be the union of bags in the subtree of~$T$ rooted a~$t$.
		
		\proofpara{Table Definition}
		Given a bag~$t$, a set of taxa~$A \subseteq Q_t$, and an integer~$s$,
		a set~$Y\subseteq V_t$ is \emph{$(t,A,s)$-feasible}, if
		\begin{enumerate}
			\item[(T1)]\label{it:color}$A$ is the subset of~$Y$ in~$Q_t$; formally~$A = Y \cap Q_t$. 
			\item[(T2)]\label{it:safeness}$Y$ contains all prey of~$Y$ in~$V_t$; formally~$\prey{Y} \cap V_t = Y$.
			\item[(T3)]\label{it:size}The size of $Y$ is $s$; formally~$|Y|=s$.
		\end{enumerate}
		
		Let~$\SSS_{t,A,s}$ be the set of $(t,A,s)$-feasible sets.
		In table entry~$\DP[t,A,s]$\todos{Remove $A$ in Food-webs. "In table entry DP[t, A, R, G, B, s], we store the largest diversity"},
		store~$\max_{Y\in \SSS_{t,A,s}} \PD(Y)$.
		Let $r$ be the root of the tree-decomposition $T$.
		Then, $\DP[r,\emptyset,k]$ stores the diversity of a solution for \Instance.
		So, return~\yes if $\DP[r,\emptyset,k]\ge D$, and \no otherwise.

		\proofpara{Leaf Node}
		For a leaf~$t$ of~$T$ the bags~$Q_t$ and $V_t$ are empty. We store
		\begin{eqnarray} \label{tw:leaf}
			\DP[t,\emptyset,0] &=& 0.
		\end{eqnarray}
		For all other values, we store $\DP[t,R,G,B,s] = -\infty$.
		
		\Recc{tw:leaf} is correct by definition.

		\proofpara{Introduce Node}
		Let~$t$ be an \emph{introduce node}, that is,~$t$ has a single child~$t'$ with~$Q_t = Q_{t'} \cup \{v\}$.
		
		If~$v\in A$ and $\prey{v} \cap Q_t \subseteq A$,
		store~$\DP[t,A,s] = \DP[t',A\setminus \{v\},s] + \PD(v)$.
		
		If~$v\not\in A$ and $\predators{v} \cap A = \emptyset$,
		store~$\DP[t,A,s] = \DP[t',A,s]$.
		
		Otherwise, if~$v\in A$ and $(\prey{v} \cap Q_t)\setminus A \ne \emptyset$, or
		if~$v\not\in A$ and $\predators{v}\cap A \ne \emptyset$,
		then store~$\DP[t,A,s] = -\infty$.

		$v$ can only be added to~$A$ if all prey are in~$A$.
		Likewise, if~$v$ is not added to $A$, then no predator can be in~$A$.

		\proofpara{Forget Node}
		Let~$t$ be a \emph{forget node}, that is,~$t$ has a single child~$t'$ and~$Q_t = Q_{t'} \setminus \{v\}$.
		We store
		\begin{eqnarray} \label{tw:forget}
			\DP[t,A,s] &=& \max\{
			\DP[t',A\cup\{v\},s];
			\DP[t',A,s]
			\}.
		\end{eqnarray}

		Define sets~$\SSS_{t,A,s}^{v} := \{ Y\in \SSS_{t,A,s} \mid v\in Y \}$ and~$\SSS_{t,A,s}^{-v} := \{ Y\in \SSS_{t,A,s} \mid v\not\in Y \}$.
		The correctness of \Recc{tw:forget} follows from the observation that~$\SSS_{t,A,s}^{v}$ and~$\SSS_{t,A,s}^{v}$ are a disjoint union of~$\SSS_{t,A,s}$; and that~$\DP[t',A\cup\{v\},s] = \max_{Y\in \SSS_{t,A,s}^{v}} \PD(Y)$, $\DP[t',A,s] = \max_{Y\in \SSS_{t,A,s}^{-v}} \PD(Y)$, and~$\DP[t,A,s] = \max_{Y\in \SSS_{t,A,s}} \PD(Y)$.

		\proofpara{Join Node}
		Let~$t$ be a \emph{join node}, that is,~$t$ has two children~$t_1$ and~$t_2$ with~$Q_t = Q_{t_1} = Q_{t_2}$.
		We store
		\begin{eqnarray} \label{tw:join}
			\DP[t,A,s] &=& \max_{s' \in [s-|A|]_0} \DP[t_1,A,|A|+s'] + \DP[t_2,A,|A|+s-s'] - \PD(A).
		\end{eqnarray}\todos{Add $|A|$ also in 1-PDD paper.}
		
		The correctness of \Recc{tw:join} follows from the fact that there are no edges between~$V_{t_1} \setminus Q_t$ and~$V_{t_2} \setminus Q_t$.
		Because~\Tree is a star, we can simply add the phylogenetic diversities together.

		\proofpara{Running Time}
		A tree decomposition contains $\Oh(n)$ nodes, thus the table contains $\Oh(2^\tw \cdot nk)$ entries.
		Any node can be computed in time linear in~$k$ and~$\tw$.
		Therefore, the overall running time is~$\Oh(2^\tw \tw \cdot nk^2)$.
	\end{proof}
}

\section{Discussion}
\label{sec:discussion}
In this paper, we defined \wPDD, a problem considering weighted food webs in the context of phylogenetic diversity maximization, as well as three special cases, \rwPDD, \fPDD, and \hPDD.
We analyzed these problems in the light of parameterized complexity for structural parameters of the food web
and presented several \XP-algorithms for \rwsPDD and several \FPT-algorithms for \fsPDD.
It is a somewhat surprising observation that for
\ifJournal
the considered parameters categorizing the structure of the food web,
\else
these parameters,
\fi
\fPDD and \fsPDD have the same complexity as \PDD and \sPDD.

It remains open whether \hPDD can be solved in polynomial time on instances where the food web is a clique and whether some of the presented \XP-algorithms
for the vertex cover number, distance to cluster, or treewidth of the food web
can be improved to \FPT-algorithms.

Some biological applications consider species interaction that generalizes one-on-one interactions~\cite{battiston2020networks}, which may be represented with a hypergraph~\cite{golubski2016ecological}.
We wonder how such interactions could be modeled in the context of maximization of phylogenetic diversity and whether such problems can be solved efficiently.

Another recent line of research is defining phylogenetic diversity in phylogenetic networks~\cite{WickeFischer2018,bordewich2022complexity,MAPPD,MaxNPD,AVGTree}.
So far, these concepts are considered without considering biological interactions.
We expect a combination of these concepts to result in very hard problems, as \PDD is already hard if the phylogenetic tree and the food web are elementary trees and most definitions of phylogenetic diversity for networks are already hard on easy network structures.
Yet, future research may identify special cases where efficient algorithms are feasible.%
\ifJournal
\footnote{Shortly after this paper has been written, Jones and Schestag presented several \FPT algorithms and a full complexity dichotomy for phylogenetic diversity on networks measured by the \textit{all-paths-PD} measure and considering ecological constraints with \viable and \fviable sets of taxa~\cite{MAPPDviability}.}
\fi

\thispagestyle{empty}
\bibliographystyle{abbrv}
\ifArXiv
\else
\bibliography{ref}
\fi

\setcounter{section}{0}
\renewcommand\thesection{A.\arabic{section}}

\renewcommand\thesection{A}
\section{Appendix}
\renewcommand\thesection{A.\arabic{section}}
\setcounter{section}{0}

\section{Proof of \Cref{lem:sub-solutions}}
\renewcommand\thelemma{\ref{lem:sub-solutions}}
\lemSubSolutions{}
\ProofLemSubSolutions{}

\section{Proof of \Cref{thm:weighted-PDD}}
\renewcommand\thetheorem{\ref{thm:weighted-PDD}}
\thmWeighted{}
\ProofThmWeighted{\CorrectnessThmWeighted}

\section{Proof of \Cref{lem:h-vc}}
\renewcommand\thelemma{\ref{lem:h-vc}}
\lemhHVC{}
\ProofLemhHVC{}

\ifJournal
\section{Proof of \Cref{cor:max-leaf-number}}
\renewcommand\thelemma{\ref{cor:max-leaf-number}}
\corMaxLeaf{}
\ProofCorMaxLeaf
\fi

\section{Proof of \Cref{thm:h-vc}}
\renewcommand\thetheorem{\ref{thm:h-vc}}
\thmhHVC{}
\ProofThmhHVC{}

\section{Proof of \Cref{thm:h-cluster}}
\renewcommand\thetheorem{\ref{thm:h-cluster}}
\thmHCluster{}
\ProofThmHCluster{}

\section{Proof of \Cref{thm:h-tw}}
\renewcommand\thetheorem{\ref{thm:h-tw}}
\thmHTW{}
\ProofThmHTW{}

\section{Proof of \Cref{lem:f-co-cluster}}
\renewcommand\thelemma{\ref{lem:f-co-cluster}}
\lemfCC{}
\ProofLemfCC{\CorrectnessLemfCC}

\section{Proof of \Cref{thm:f-tw}}
\renewcommand\thetheorem{\ref{thm:f-tw}}
\thmfTW{}
\ProofThmfTW{}


\end{document}

%% file: commands.tex
\usepackage[utf8]{inputenc}
\usepackage{amsmath}
\usepackage{amssymb}
\usepackage{amsthm}
\usepackage{etex}
\usepackage{enumerate}
\usepackage{dsfont}
\usepackage{thmtools}
\usepackage{mathtools}
\usepackage{graphicx}
\usepackage{multirow}
\usepackage[table]{xcolor}
\usepackage{placeins}
\usepackage{nicefrac}
\usepackage{tabularx}
\usepackage{algorithm2e}
\usepackage[T1]{fontenc}
\usepackage{etoolbox}
\usepackage{hyperref}
\usepackage{paralist}
\usepackage{booktabs}
\hypersetup{
	colorlinks,
	linkcolor={red!50!black},
	citecolor={blue!90!black!80},
	urlcolor={blue!80!black}
}
\usepackage{cleveref}

\usepackage[disable]{todonotes}
\usepackage{lineno}

\usepackage{tikz}
\usetikzlibrary{graphs}
\usetikzlibrary{arrows.meta,shapes,positioning,calc,decorations.pathreplacing,decorations.markings,decorations.pathmorphing,patterns}
\pgfdeclarelayer{background2}
\pgfdeclarelayer{background}
\pgfdeclarelayer{foreground}
\pgfsetlayers{background2,background,main,foreground}

\usepackage[edges]{forest}

\tikzstyle{bold}=[draw, line width=2pt]
\tikzstyle{optional}=[dashed]
\tikzstyle{path}=[decorate, decoration={snake, amplitude=.6mm}]

\tikzstyle{small}=[inner sep=2pt]
\tikzstyle{tiny}=[inner sep=1.7pt]
\tikzstyle{textnode}=[inner sep=0pt]

\tikzstyle{triangle}=[draw, regular polygon, regular polygon sides=3]
\tikzstyle{vertex}=[circle, draw, fill=white]
\tikzstyle{reti}=[vertex, fill=black]
\tikzstyle{leaf}=[vertex, rectangle]
\tikzstyle{leaf2}=[vertex, regular polygon, regular polygon sides=3]

\tikzstyle{smallvertex}=[vertex, small]
\tikzstyle{smallleaf}=[leaf, inner sep=3.3pt]
\tikzstyle{smallleaf2}=[leaf2, inner sep=1.7pt]
\tikzstyle{smalltriangle}=[triangle, inner sep=1.5pt]
\tikzstyle{smallreti}=[reti, small]

\tikzstyle{match}=[edge,line width=3pt]
\tikzstyle{edge}=[draw,-]
\tikzstyle{arc}=[draw,arrows={-Latex[length=6pt]}]
\tikzstyle{boldarc}=[draw, bold, arrows={-Latex[length=10pt]}]
\tikzstyle{revarc}=[draw, arrows={Latex[length=6pt]-}]
\tikzstyle{boldrevarc}=[draw, bold, arrows={Latex[length=10pt]-}]

\newcommand{\nextnode}[5][vertex]{\node[small#1] (#2) at ($(#3)+(#4)$) {} edge[#5] (#3);}

{\upshape\itshape}{\upshape\rmfamily}

{\upshape\itshape}{\upshape\rmfamily}
{\upshape\itshape}{\upshape\rmfamily}
\usepackage{multirow}

\newcommand{\myrowcolsnum}[1]{\rowcolors{#1}{gray!10!white}{white}}

\newcommand{\kommentar}[1]{}
\newcommand{\Oh}{\ensuremath{\mathcal{O}}}
\newcommand{\SSS}{\ensuremath{\mathcal{S}}}

\newcommand{\proofpara}[1]{\noindent\textit{#1.}}

\DeclareMathOperator{\twwithoutN}{{tw}}

\DeclareMathOperator{\off}{off}

\DeclareMathOperator{\DP}{DP}

\newcommand{\w}{{\ensuremath{\omega}}\xspace}

\newcommand\Recc[1]{Recurrence~(\ref{#1})}

\newcommand{\predators}[1]{{\ensuremath{N_{>}(#1)}}\xspace}
\newcommand{\prey}[1]{{\ensuremath{N_{<}(#1)}}\xspace}

\newcommand{\Wmax}{\ensuremath{W_{\max}}\xspace}

\newcommand{\yes}{{\normalfont{yes}}\xspace}
\newcommand{\no}{{\normalfont{no}}\xspace}

\newcommand{\Wh}[1]{{\normalfont W[#1]}\xspace}
\newcommand{\NP}{{\normalfont{NP}}\xspace}

\newcommand{\FPT}{{\normalfont{FPT}}\xspace}

\newcommand{\XP}{{\normalfont{XP}}\xspace}

\newcommand{\Instance}{{\ensuremath{\mathcal{I}}}\xspace}

\newcommand{\Tree}{{\ensuremath{\mathcal{T}}}\xspace}
\newcommand{\Food}{{\ensuremath{\mathcal{F}}}\xspace}

\newcommand{\tw}{{\ensuremath{\twwithoutN_{\Food}}}\xspace}
\newcommand{\PD}{\PDsub\Tree}
\newcommand{\PDsub}[1]{{\ensuremath{PD_{#1}}}\xspace}

\newcommand{\problemdef}[3]{
	\begin{quote}
		\normalsize\textsc{#1} \smallskip \\
		\begin{tabularx}{0.9\textwidth}{@{}l@{\hspace{3pt}}X}
			\normalsize\textbf{Input:}    & \normalsize#2 \\
			\normalsize\textbf{Question:} & \normalsize#3
		\end{tabularx}
	\end{quote}
}

\newcommand{\PROB}[1]{{{\normalfont\textsc{#1}}}\xspace}

\newcommand{\MPD}{\PROB{Max-PD}}
\newcommand{\MPDlong}{\PROB{Maximize Phylogenetic Diversity}}

\newcommand{\KP}{\PROB{Knapsack}}

\newcommand{\ILPF}{\PROB{ILP-Feasibility}}